\documentclass[12pt]{article}
\usepackage{amsmath}
\usepackage{graphicx}
\usepackage{natbib}
\usepackage{url} 
\usepackage[table,x11names]{xcolor}

\usepackage{graphicx}
\usepackage{longtable}
\usepackage{float}
\usepackage{forest}
\usepackage{setspace}
\usepackage{wrapfig}
\usepackage{soul}
\usepackage{amsmath}
\usepackage{mathtools}
\usepackage{textcomp}
\usepackage{marvosym}
\usepackage{wasysym}
\usepackage{latexsym}
\usepackage{amssymb}
\usepackage{hyperref}
\usepackage{graphicx}
\usepackage{subfig}
\usepackage{longtable}
\usepackage{float}
\usepackage{array}
\usepackage{algorithmicx}
\usepackage[plain,noend,ruled]{algorithm2e}
\usepackage{algpseudocode}

\usepackage{comment}
\newcommand{\blind}{0}

\newcommand{\norm}[1]{\left\lVert#1\right\rVert}

\newcommand{\cov}{\text{Cov}}

\newcommand{\vm}{\text{vec}}

\input{./Definitions}

\begin{document}

\def\spacingset#1{\renewcommand{\baselinestretch}%
{#1}\small\normalsize} \spacingset{1}


\if0\blind
{
  \title{\bf Regularized Parameter Estimation in Mixed Model Trace Regression}
  \author{Ian Hultman \hspace{.2cm}\\
    Department of Statistics and Actuarial Science, University of Iowa\\
    and \\
    Sanvesh Srivastava \\
    Department of Statistics and Actuarial Science, University of Iowa}
  \maketitle
} \fi

\if1\blind
{
  \bigskip
  \bigskip
  \bigskip
  \begin{center}
    {\LARGE\bf Title}
\end{center}
  \medskip
} \fi

\bigskip
\begin{abstract}
We introduce mixed model trace regression (MMTR), a mixed model linear
regression extension for scalar responses and high-dimensional matrix-valued
covariates. MMTR's fixed effects component is equivalent to trace regression,
with an element-wise lasso penalty imposed on the regression coefficients matrix
to facilitate the estimation of a sparse mean parameter. MMTR's key innovation lies in modeling the covariance structure of matrix-variate random effects as a Kronecker product of low-rank row and column covariance matrices, enabling sparse estimation of the covariance parameter through low-rank constraints. We establish identifiability conditions for the estimation of row and column covariance matrices
and use them for rank selection by applying group lasso regularization on the
columns of their respective Cholesky factors. We develop an Expectation-Maximization (EM) algorithm extension for numerically stable parameter estimation in high-dimensional applications.  MMTR achieves estimation accuracy comparable to leading regularized quasi-likelihood competitors across diverse simulation studies and attains the lowest mean square prediction error compared to its competitors on a publicly available image dataset.
\end{abstract}

\noindent%
\vfill

\newpage
\spacingset{1.75} 

\section{Introduction}
\label{sec:intro}

We introduce trace regression models that include matrix-valued random effects, referred to as \emph{mixed model trace regression} (MMTR). MMTR extends mixed model linear regression to include matrix-valued fixed and random effects covariates, but the response remains a vector of correlated observations. The parameter dimension in these models easily exceeds the number of observations, requiring specific sparsity assumptions on the model parameters. MMTR models the mean parameter  by assuming that the regression coefficients matrix is sparse. The random effects' row and column covariance matrices in MMTR are assumed to have low ranks, leading to sparse covariance parameters. We develop an EM algorithm for regularized parameter estimation in MMTR, which exploits the sparsity patterns in the parameter components for efficiency and stability in the high-dimensional regime.

Consider the MMTR setup. Let $n$ be the number of subjects or clusters, $m_i$ be the number of observations specific to cluster $i$ ($i=1, \ldots, n$), and $N = \sum_{i=1}^n m_i$ be the total number of observations. The response vector for cluster $i$ is $\yb_i \in \RR^{m_i}$, with $j$th element $y_{ij} \in \RR$. The fixed and random effects covariates matrices specific to $y_{ij}$ are $\Xb_{ij} \in \RR^{P_1 \times P_2}$ and $\Zb_{ij} \in \RR^{Q_1 \times Q_2}$ ($i=1, \ldots, n$; $j=1, \ldots, m_i$). Then, MMTR sets
\begin{align}
  \label{eq:tensr-mdl}
  y_{ij} = \tr( \Xb_{ij}^\top \Bb) +  \tr(\Zb_{ij}^\top \Ab_i)  + e_{ij},
\quad  \Ab_i \sim N_{Q_1, Q_2}(\zero, \tau^2 \Sigmab_1,  \Sigmab_2), \quad \eb_i \sim N_{m_i}(\zero, \tau^2 \Ib_{m_i}),
\end{align}
where $\tr$ is the \emph{trace} operator, $\Bb \in \RR^{P_1 \times P_2}$ is the fixed effects parameter matrix, $\Ab_i \in \RR^{Q_1 \times Q_2}$ is the random effects matrix, $N_{Q_1, Q_2}(\zero, \tau^2 \Sigmab_1, \Sigmab_2)$ is the matrix variate Gaussian distribution with row and column covariance matrices $\tau^2\Sigmab_1$ and $\Sigmab_2$, respectively, $\eb_i = (e_{i1}, \ldots, e_{im_i})^\top$ is the $i$th idiosyncratic error vector, $\Ib_{m_i}$ is an $m_i \times m_i$ identity matrix, and $\Ab_i, \eb_i$ are mutually independent for every $i$. The model parameters are $\Bb, \Sigmab_1, \Sigmab_2,$ and $ \tau^2$.

Mixed model linear and trace regression models are special cases of MMTR in \eqref{eq:tensr-mdl}. MMTR reduces to mixed model linear regression if $P_2 = Q_2 = 1$. This model is widely used for analyzing repeated measures and clustered data, but its application is limited to vector-valued covariates \citep{verbeke1997linear}. If $\Ab_i = \zero$ and $m_{i}=1$ for every $i$ in \eqref{eq:tensr-mdl}, then MMTR is equivalent to the trace regression model \citep{Fanetal19}. This model accommodates matrix-valued covariates but fails to model the correlation in the responses. MMTR integrates the advantages of these two frameworks and enables realistic models for repeated measures and clustered data with matrix covariates and scalar responses.

\subsection{Prior Literature}
%

The vectorization of covariates reduces MMTR to a ``structured'' mixed model linear regression. If $P=P_{1}P_{2}$ and $Q=Q_{1}Q_{2}$, then  \eqref{eq:tensr-mdl} is equivalent to a mixed model with fixed and random effects covariates $\vm(\Xb_{{ij}}) \in \RR^{P}$ and $\vm(\Zb_{ij}) \in \RR^{Q}$, $P$-dimensional regression coefficient $\vm(\Bb)$, and $Q$-dimensional random effects $\vm(\Ab_{i})$ with a $Q \times Q$ covariance matrix $\tau^{2} \Sigmab_{2} \otimes \Sigmab_{1}$, where $\vm$ is the column-wise vectorization of a matrix and  $\otimes$ is the Kronecker product. Fitting these structured mixed models using classical algorithms, such as those in \texttt{lme4} \citep{Bat15}, has poor empirical performance for two main reasons. First, the estimation algorithms do not exploit the sparse and the separable matrix-variate structures of $\Bb$ and $\tau^{2} \Sigmab_{2} \otimes \Sigmab_{1}$, respectively. 
Second, they are prohibitively slow because the parameter dimension grows rapidly as $O(P + Q^2)$.

High-dimensional mixed model extensions address these issues through penalization but fail to exploit the separable covariance structures. The main idea of these methods is to define a quasi-likelihood by replacing the covariance matrix with a proxy matrix \citep{FanLi12,Huietal17,Braetal20,Lietal22}. A sparse estimate of $\Bb$ is obtained using lasso-type penalties, which are optimal under various high-dimensional asymptotic regimes. The estimate of the random effects covariance matrix is  only available when $Q < \min_{i} m_{i}$, a condition often violated in practice \citep{Lietal22}. Furthermore, even when the estimate exists, it does not have the separable structure of the true covariance matrix.

The literature on separable covariance estimation addresses such problems. Consider a modification of \eqref{eq:tensr-mdl} with $\Bb = \zero$ and $ \tau^{2}=1$ that sets $\Yb_{i} = \Ab_{i} \in \RR^{Q_{1} \times Q_{2}}$ for $i=1, \ldots, n$. In this model, $\cov \{\text{vec}(\Ab_{i})\}$ has the separable form $\Sigmab_{2} \otimes \Sigmab_{1}$, where $\Sigmab_{1}$ and $\Sigmab_{2}$ are defined in \eqref{eq:tensr-mdl}. In the low dimensional setting, the two parameters are estimated using a ``flip-flop'' algorithm that estimates $\Sigmab_{2}$ given $\Sigmab_{1}$ and vice versa \citep{Dut99,LuZim05,Srietal08}. \citet{Hof11} develops tensor-variate extensions of these algorithms; however, all these approaches assume that $Q \ll n$. Extending methods from the literature on low-rank covariance estimation, \citet{Zhaetal23} estimate $\cov \{\text{vec}(\Ab_{i})\}$ via regularized banded estimates of $\Sigmab_{1}$ and $\Sigmab_{2}$. Unlike the setup in these models, the random effects $\Ab_{i}$'s are unobserved in \eqref{eq:tensr-mdl}, making it impossible to apply these algorithms directly for parameter estimation in MMTR.

Tensor regression directly handles matrix-variate covariates but assumes $\Ab_i = \zero$ and $m_i = 1$ for every $i$ in \eqref{eq:tensr-mdl}. The main focus is on estimating $\Bb$ under sparsity inducing penalties. If the correlation in $\yb_i$ induced by $\Ab_i$ in \eqref{eq:tensr-mdl} is ignored, then existing trace regression algorithms estimate low-rank, row-sparse, or column-sparse estimates of $\Bb$ in \eqref{eq:tensr-mdl}  using nuclear norm and group lasso penalties \citep{zhao2017,slawski2015regularization,Fanetal19}. While such estimates remain under-studied, similar estimates, which assume the independence of $\yb_i$ entries, in mixed model linear regression have poor inferential and predictive performance \citep{Huietal21,Lietal22}. We conjecture that such results extend to the MMTR due to the equivalence between the two model classes.

MMTR belongs to the class of tensor mixed models, where the responses and covariates are structured as tensors. The simplest models in this class have no random effects and estimate the regression coefficients under low rank and sparse constraints \citep{zhou2013,zhou2014regularized}. There is limited literature on tensor mixed models that jointly estimate mean and covariance parameters. \citet{Yueetal20} develop one such model, but it excludes random effects covariates. The random effects tensor has a separable structure in this model and the estimation algorithm is a variant of the flip-flop algorithm; however, their approach is only applicable when $n \ll Q$ and $\Zb_{ij} = \Ib_{Q_{1}}$ for every $i$ and $j$, implying that this method cannot be used for parameter estimation in MMTR.

MMTR is closely related to tensor mixed models based on generalized estimating
equations (GEE), which are tuned for modeling longitudinal imaging data
\citep{Zhaetal19}. The GEE-based model has scalar responses and tensor
covariates. For two dimensional tensors, the GEE-based models and MMTR have the same mean parametrization; however, the GEE-based model assumes that the regression coefficients have a low tensor rank, whereas MMTR assumes they are sparse. The random effects terms are absent in the GEE-based model, which replaces them with sample-specific marginal ``working'' covariance matrices. The GEE-based model requires selection of the rank, working covariance matrix form, and penalty tuning parameters for parameter estimation using a minorization maximization (MM) algorithm. In contrast, MMTR's parameter estimation algorithm only requires two tuning parameters that determine the sparsity of $\Bb$ and ranks of $\Sigmab_1$ and $\Sigmab_2$.

\subsection{Our Contributions}
\label{sec:our-cont}

MMTR employs a regularized EM for estimating the parameters in \eqref{eq:tensr-mdl}. Let $\Lb_1$ and $\Lb_2$ be two matrices such that $\Sigmab_1 = \Lb_1 \Lb_1^\top$ and $\Sigmab_2 = \Lb_2 \Lb_2^\top$. Then, the parameters in \eqref{eq:tensr-mdl} are $\thetab = \{\Bb, \Lb_1, \Lb_2, \tau^2\}$. The E step treats the random effects $\Ab_1, \ldots, \Ab_n$ as missing data and replaces the log-likelihood with a minorizer, which is analytically tractable due to the set up in \eqref{eq:tensr-mdl}. The M step maximizes the E step minorizer through a series of closed-form conditional maximizations; however, this approach requires $P + Q_1^{2} + Q_2^{2} \ll n$ for numerical stability. For vector-valued covariates, this algorithm reduces to the EM algorithm in \citet{Dyk00} for parameter estimation in mixed model linear regression.

In the high-dimensional settings where $n \ll P + Q_1^2 + Q_2^2$, the previous EM algorithm requires regularization. We assume that $\Bb$ is sparse and that $\Sigmab_1$ and $\Sigmab_2$ are approximately low rank; that is, $\Sigmab_1 \approx \Lb_1 \Lb_1^\top$ and $\Sigmab_2 \approx \Lb_2 \Lb_2^\top$, where $\Lb_k \in \RR^{Q_k \times  S_k}$, and $S_k \ll Q_k$ for  $k=1, 2$, reducing the parameter dimension from $O(P_1 P_2 + Q_1^2 + Q_2^2)$ to $O(P_1 P_2 + Q_1S_1+ Q_2S_2)$. The $\Sigmab_k$ estimates depend on $S_k$ choice, so we use group lasso penalties on the columns of $\Lb_k$ for a data-driven choice of $S_k$. The low-rank factorization of $\Sigmab_k$ resembles a factor-analytic structure, so we set $S_k = O(\log Q_k)$ initially \citep{RocGeo16,Sri17}. This idea extends the low-rank random effects in mixed model linear regression to low-rank matrix-variate random effects in MMTR \citep{james2000principal,Hei24}. This setup implies that the maximization of the E step objective reduces to a series of regularized least squares problems, enabling estimation of $\Bb, \tau^2, \Lb_1$ and $\Lb_2$ via \texttt{scalreg} \citep{scalreg_lib} and \texttt{gglasso} \citep{gglasso_lib} \emph{R} packages; see Section \ref{sec:em-algo} for the details.

Our main contributions are threefold. First, MMTR is a novel extension of mixed
models for matrix-valued covariates. We establish conditions for the separate
estimation of $\Sigmab_1$ and $\Sigmab_2$, given that only $\Sigmab_2 \otimes
\Sigmab_1$  is identified in \eqref{eq:tensr-mdl}; see Section
\ref{sec:matrix-lme-mdl}. Second, random effects $\Ab_i$'s in \eqref{eq:tensr-mdl}
have a low-rank separable covariance structure, which is a natural extension of
separable covariance arrays to high-dimensional settings \citep{Hof11}. MMTR is
also related to separable factor analysis (SFA) in that SFA adds diagonal
matrices with positive entries to $\Sigmab_1$ and $\Sigmab_2$, implying that
MMTR's covariance structure is more parsimonious than SFA when $\min(P_{1},
P_{2})$ is large  \citep{FosHof14}; see Section \ref{sec:mmtr_notation}.
Finally, we develop an EM algorithm for efficient parameter regularized
estimation and data-driven choice of the ranks of $\Sigmab_1$ and $\Sigmab_2$; see
Section \ref{sec:em-algo}. MMTR achieves lower estimation and prediction errors than its high-dimensional
mixed model competitors across diverse simulated and real data analyses.

\section{Mixed Model Trace Regression}
\label{sec:matrix-lme}

\subsection{Model}
\label{sec:matrix-lme-mdl}
Consider the parameter expanded form of \eqref{eq:tensr-mdl}. Let $\Lb_k \in \RR^{Q_k \times S_k}$ be the square root matrix factor of $\Sigmab_k$ for $k=1, 2$. If $\Lb_k$ has full rank for $k = 1, 2$, then $Q_k = S_k$ and $\Sigmab_k = \Lb_k \Lb_k^\top$. MMTR imposes a low-rank structure on $\Sigmab_1$  and $\Sigmab_2$ by assuming that $S_1 \ll Q_1$ and $S_2 \ll Q_2$, so that $\Lb_1 \Lb_1^\top \approx \Sigmab_1$ and $\Lb_2 \Lb_2^\top \approx \Sigmab_2$. The random effects are expressed using $\Lb_1$ and $\Lb_2$ as
\begin{align}
  \label{eq:rnd-eff}
  \Ab_i = \Lb_1 \Cb_i \Lb_2^\top, \quad  \Cb_i \sim N_{S_1, S_2}(\zero, \tau^2 \Ib_{S_1},  \Ib_{S_2}),
\end{align}
where $\Cb_i$ is a $S_1 \times S_2$ matrix of independent normal random variables with mean 0 and variance $\tau^2$. Based on these assumptions, the parameter expanded form of MMTR in \eqref{eq:tensr-mdl} is
\begin{align}
  \label{eq:px-tensr-mdl}
  y_{ij} = \tr( \Xb_{ij}^\top \Bb) +  \tr(\Zb_{ij}^\top \Lb_1 \Cb_i \Lb_2^\top)  + e_{ij},
\quad e_{ij} \sim N(0, \tau^2).
\end{align}
The model parameters are $\thetab = (\Bb, \Lb_1, \Lb_2, \tau^2)$, and $\Lb_k$ for $k=1, 2$ is non-identified because replacing $\Lb_k$ with $\Lb_k \Ob_k$ for any orthonormal matrix $\Ob_k$ leaves the decomposition of $\Sigmab_k$ in \eqref{eq:rnd-eff} unchanged. We do not impose  identifiability conditions on $\Lb_1$ and $\Lb_2$ because they do not impact the estimates of $\Sigmab_{1}$ and $\Sigmab_{2}$.

MMTR in \eqref{eq:px-tensr-mdl} reduces to the trace regression model when $m_i=1$ and $\Ab_i=\zero$ for every $i=1, \ldots, n$ \citep{slawski2015regularization,zhao2017,Fanetal19}.  These methods estimate $\Bb$ under different regularization schemes using the lasso-type and nuclear norm penalties. For example, $\| \Bb \|_1 = \sum_{ij} | b_{ij} |$, $\| \Bb \|_{1,2} = \sum_{i=1}^{P_1} \| \bb_{i:} \|_2 $, and $\| \Bb \|_{2, 1} = \sum_{j=1}^{P_2} \| \bb_{:j} \|_2$, where $b_{ij}$, $\bb_{i:}$, and $\bb_{:j}$ are the $(i,j)$th entry, $i$th row, and $j$th column of $\Bb$. The application of these three penalties results in sparse, row sparse, and column sparse $\Bb$ estimates, which are efficiently estimated using the alternating direction method of multipliers or projected gradient methods. If $\Lb_1, \Lb_2, \tau^2$ are known, then we can use these algorithms for sparse estimation of $\Bb$. MMTR extensions with $\| \Bb \|_{1,2} $ and $\| \Bb \|_{2, 1}$ penalties can leverage these estimation algorithms; however, we employ only the $\| \Bb \|_1$ penalty for simplicity, enabling the direct application of existing algorithms designed for scaled sparse regression \citep{sun2012scaled}.

The covariance matrices $\Sigmab_1$ and $ \Sigmab_2$ in \eqref{eq:tensr-mdl} are identified up to a constant. For $i = 1, \ldots, n$, $\Ab_i$ in \eqref{eq:tensr-mdl} has a separable covariance array; that is, $\cov \{\text{vec}(\Ab_i)\} = \tau^2 \Sigmab_2 \otimes \Sigmab_1$ \citep{Dut99,LuZim05,Hof11}. The covariance array remains unchanged if $\Sigmab_1$ and $ \Sigmab_2$ are modified to $c' \Sigmab_1$ and $\Sigmab_2 / c'$ for any $c' > 0$. If $\sigma^2$ is the positive ``overall variance'' parameter, then $\Sigmab_1$ and $ \Sigmab_2$ are identified by assuming that $\det(\Sigmab_1)= \det(\Sigmab_2) = 1$ and $\cov\{\text{vec}(\Ab_i)\} = \tau^2 \sigma^2 \Sigmab_2 \otimes \Sigmab_1$ \citep{GerHof15}. This condition is inapplicable for identifying  $\Sigmab_1$ and $ \Sigmab_2$ in \eqref{eq:tensr-mdl} because they are low-rank matrices.

We enforce the identifiability of $\Sigmab_1$ and $\Sigmab_2$ based on $\Lb_2$. Let $(\Mb)_{ij}$ denote the $(i,j)$th element of some matrix $\Mb$. A popular identifiability condition assumes that $(\Sigmab_2)_{11} = 1$, implying that the first $Q_1 \times Q_1$ block of $\Sigmab_2 \otimes \Sigmab_1$ equals $\Sigmab_1$ \citep{Dut99}. We can extend this condition to the case when $(\Sigmab_2)_{jj} = 1$, which results in the $j$th diagonal $Q_1 \times Q_1$ block of $\Sigmab_2 \otimes \Sigmab_1$ equaling $\Sigmab_1$. This condition is useful when $\Sigmab_1$ and $\Sigmab_2$ have full ranks and are estimated directly; however, MMTR estimates $\Sigmab_1$ and $\Sigmab_2$ via $\Lb_1$ and $\Lb_2$, so this condition requires modification. The following proposition restates this identifiability condition using $\Lb_{2}$, which is suited for high-dimensional MMTR applications.
\begin{proposition}\label{prop1}
  Assume that $(\Lb_2)_{j1} = 1$ and $(\Lb_2)_{j2} = \cdots = (\Lb_2)_{jS_2} = 0$. Then, $(\Sigmab_2)_{jj} = 1$ and the $j$th diagonal $Q_1 \times Q_1$ block of $\Sigmab_2 \otimes \Sigmab_1$ is $\Sigmab_1$.
\end{proposition}
The proof of this proposition is given in the supplementary material along with other proofs. Proposition \ref{prop1} does not impose any assumption on the ranks of $\Sigmab_1$ and $ \Sigmab_2$, so it is applicable when $\Sigmab_1$ or $\Sigmab_2$ are rank-deficient. Furthermore, the following proposition shows that we can convert any square-root covariance matrix $\Lb_2$ into the form described in Proposition \ref{prop1}
\begin{proposition}\label{prop2}
  For any matrix $\Lb \in \RR^{Q \times S}$, there exists an orthonormal matrix $\Qb_j \in \RR^{S \times S}$ and constant $c$ such that $c \Lb \Qb_j^\top$ is a matrix whose $j$th row is the first standard basis vector in $\RR^S$.
\end{proposition}

Proposition \ref{prop2} states that $\Sigmab_1$ and $\Sigmab_2$ are identified if $\Sigmab_2$ is scaled so that any one of its diagonal values equals 1. We use Proposition \ref{prop2} to evaluate the estimation accuracy of our estimation algorithm.
We compare our estimates $\hat{\Sigmab}_1$ and $\hat{\Sigmab}_2$ of the true parameters $\Sigmab_1$ and $\Sigmab_2$ in the simulations by scaling $\hat{\Sigmab}_1$, $\hat{\Sigmab}_2$, $\Sigmab_1$, and $\Sigmab_2$ so that their respective largest diagonal entries are 1. We also use Proposition \ref{prop2} to scale the $ {\Sigmab}_1$ and $ {\Sigmab}_2$ estimates at the end of each iteration to further improve the stability of the estimation algorithm.

\subsection{Reduction to a Structured Mixed Model}
\label{sec:mmtr_notation}

The reduction of MMTR to a structured mixed model linear regression requires the following definitions. Let $\bb = \text{vec}(\Bb)$, $\xb_{ij} = \text{vec}(\Xb_{ij})$, $\Zb_{ij(1)} = \Zb_{ij}$, $\Zb_{ij(2)} = \Zb_{ij}^\top$, $\zb_{ij(1)} = \text{vec}(\Zb_{ij(1)})$, $\zb_{ij(2)} = \text{vec}(\Zb_{ij(2)})$, $\ab_{i(1)} = \text{vec}(\Ab_i)$, $\ab_{i(2)} = \text{vec}(\Ab_i^\top)$, $\cb_{i(1)} = \text{vec}(\Cb_i)$, $\cb_{i(2)} = \text{vec}(\Cb_i^\top)$, $\lb_1 = \text{vec}(\Lb_1)$, $\lb_2 = \text{vec}(\Lb_2)$. Then, identities relating $\tr$, vec, and $\otimes$ operators imply that  $\Ab_i = \Lb_1 \Cb_i \Lb_2^\top$ is equivalent to $\ab_{i(1)} = \Lb_{(1)} \cb_{(1)}$ and $\ab_{i(2)} = \Lb_{(2)} \cb_{i(2)}$, where $\Lb_{(1)} = \Lb_2 \otimes \Lb_1$ and $\Lb_{(2)} = \Lb_1 \otimes \Lb_2$; see \citet{seber_kronecker}. The vectors $\ab_{i(1)}$'s and $\ab_{i(2)}$'s are used in estimating $\Lb_{1}$ given $\Lb_{2}$ and $\Lb_{2}$ given $\Lb_{1}$, respectively.

MMTR in \eqref{eq:px-tensr-mdl} is equivalent to a mixed model linear regression with the mean parameter $\bb$ and random effects covariance $\tau^2 \Lb_2\Lb_2^\top \otimes \Lb_1 \Lb_1^\top$. Using \eqref{eq:px-tensr-mdl} and previous identities,
\begin{align}
  \label{eq:e1}
  y_{ij} &=  \xb^\top_{ij} \bb + \zb_{ij(1)}^\top (\Lb_2 \otimes \Lb_1) \cb_{i(1)} + e_{ij}, \quad
  \yb_i =  \Xb_{i} \bb + \Zb_{i(1)} (\Lb_2 \otimes \Lb_1) \cb_{i(1)} + \eb_{i},
\end{align}
for $i=1, \ldots, n$ and $j = 1, \ldots, m_i$, where $\Xb_i$ is the matrix whose $j$th row is $\xb_{ij}^\top$ and $\Zb_{i(1)}$ is the matrix whose $j$th row is $\zb_{ij(1)}^\top$.  The first equation uses $\tr(\Zb_{ij}^\top \Lb_1 \Cb_i \Lb_2^\top) = \zb_{ij(1)}^\top \text{vec}(\Lb_1 \Cb_i \Lb_2^\top) = \zb_{ij(1)}^\top (\Lb_2 \otimes \Lb_1) \cb_{i(1)}$, which implies that the covariance matrix of random effects is $\cov \{(\Lb_2 \otimes \Lb_1) \cb_{i(1)}\} = \tau^2 \Lb_2\Lb_2^\top \otimes \Lb_1 \Lb_1^\top$. If $\Sigmab_1$ and $\Sigmab_2$ are full-rank and $n$ and $m_1, \ldots, m_n$ are sufficiently large, then we estimate $\Bb= \text{vec}^{-1}(\bb)$ and $\Sigmab=\Sigmab_2 \otimes \Sigmab_1$ using existing software such as \texttt{lme4}.
If $\| \cdot \|_F$ is the Frobenius norm, then the $ \Sigmab_1$  and $ \Sigmab_2$ estimates minimize $\|\hat \Sigmab - \Sigmab_2 \otimes \Sigmab_1\|_F$ and are obtained using the singular value decomposition (SVD) of a matrix formed by reordering $\hat{\Sigmab}$ entries \citep{vanLoan00}. The complexity of computing $\hat \Sigmab_1$  and $\hat \Sigmab_2$ using this method scales as $O(Q_1^3Q_2^3)$, which leads to inefficiency in practice even for small $Q_1$ and $Q_2$. Furthermore, high-dimensional extensions of  mixed models face similar issues.



The quasi-likelihood approaches bypass these problems by replacing $\Sigmab_{1}$ and $ \Sigmab_{2}$ with proxy matrices that have certain asymptotic properties. The resulting loss for the estimation of $\bb$ is equivalent to that of weighted least squares regression. The penalized $\bb$ estimate is obtained using a lasso-type penalty on $\bb$ that has optimal theoretical properties depending on the choice of proxy matrices and tuning parameter \citep{FanLi12,Lietal22}. The estimation of $\Sigmab$ is still challenging in that it requires that $Q_1Q_2 < \min_i m_i$, limiting their practical applications  \citep{Lietal22}.

Due to these reasons, the structured mixed model for MMTR has two forms that enable estimation of $\Lb_1$ given $\Lb_2$ and vice versa. These two models facilitate the estimation of $\Sigmab_1$ and $\Sigmab_2$ without any previous restrictions.
For $i=1, \ldots, n$ and $j = 1, \ldots, m_i$, $\tr(\Xb_{ij}^\top\Bb) = \xb^\top_{ij} \bb$, $\tr(\Zb_{ij}^\top \Ab_i) = \zb_{ij(1)}^\top \ab_{i(1)} = \zb_{ij(1)}^\top \Lb_{(1)} \cb_{(1)}$, and $\tr(\Zb_{ij} \Ab_i^\top) = \zb_{ij(2)}^\top \ab_{i(2)} = \zb_{ij(2)}^\top \Lb_{(2)} \cb_{(2)}$. Substituting these identities in \eqref{eq:px-tensr-mdl} gives
\begin{align}\label{eq:3}
  y_{ij} &=  \xb^\top_{ij} \bb + \zb_{ij(1)}^\top \Lb_{(1)} \cb_{i(1)} + e_{ij}
           = \xb^\top_{ij} \bb + \zb_{ij(2)}^\top \Lb_{(2)} \cb_{i(2)} + e_{ij} ,
           \quad e_{ij} \sim N(0, \tau^2),
\end{align}
where $\cb_{i(1)}$ and $ \cb_{i(2)}$ are distributed as $N(0, \tau^2 \Ib_{S_1 S_2})$. Combining \eqref{eq:3} for $j = 1, \ldots, m_i$ gives
\begin{align}\label{eq:4}
  \yb_{i} &= \Xb_{i} \bb + \Zb_{i(1)} \Lb_{(1)} \cb_{i(1)} + \eb_{i}, \quad
  \yb_i = \Xb_{i} \bb + \Zb_{i(2)} \Lb_{(2)} \cb_{i(2)} + \eb_{i}, \quad \eb_{i} \sim N(0, \tau^2 \Ib_{m_i}),
\end{align}
for $i=1, \ldots, n,$, where the $j$th row of $\Xb_i$ is $\xb_{ij}^\top$, the $j$th row of $\Zb_{i(1)}$ is $\zb_{ij(1)}^\top$, and the $j$th row of $\Zb_{i(2)}$ is $\zb_{ij(2)}^\top$. The two models in \eqref{eq:4} are used for estimating $\Lb_1$ given $\Lb_2$ and vice versa.

We regularize $\Bb, \Lb_1, \Lb_2$ to ensure numerically stable parameter estimation in the high-dimensional settings. As noted earlier, the trace regression literature offers a range of penalty options for $\Bb$, but we use the lasso penalty  because it facilitates efficient estimation via the \texttt{scalreg} package. An additional advantage of using \texttt{scalreg} is that its $\tau^2$ estimate has a smaller bias compared to other alternatives, including \texttt{glmnet} \citep{FriHasTib10}. The penalties on $\Lb_1$ and $ \Lb_2$ are chosen to avoid the need for the exact specification of $S_1$  and $S_2$, which are typically unknown. The $\Lb_1$ and $\Lb_2$ matrices define low-rank row and column covariance matrices in that $\Sigmab_{k} \approx \Lb_{k} \Lb_{k}^{\top}$ and $S_{k} \ll Q_{k}$ for $k=1, 2$. Leveraging similar results in factor analysis that bypass the number of latent factors specification, we set $S_1 = O(\log Q_1)$ and $S_2 = O(\log Q_2)$ and use the group lasso penalty on the columns of $\Lb_1$ and $ \Lb_2$ \citep{RocGeo16,Sri17}. This penalty offers a data-driven approach for estimating $S_1$ and $S_2$, where their estimates correspond to the highest column index with nonzero entries.

\section{Regularized Parameter Estimation}
\label{sec:em-algo}

We employ a regularized alternating expected-conditional maximization (AECM) algorithm with three cycles for parameter estimation. This algorithm is an EM extension that estimates $(\bb, \tau^{2})$, $\Lb_{1}$, and $\Lb_{2}$ across its three cycles, with each cycle conditioning on the remaining parameters. The first cycle estimates $(\bb, \tau^{2})$ using regularized weighted least squares regression. The second and third cycles treat $\cb_{i(2)}$'s and $\cb_{i(1)}$'s in \eqref{eq:4} as missing data and estimate $\Lb_{1}$ and $\Lb_{2}$, respectively. The derivation of the parameter updates is in the supplementary material.

Consider the $(t+1)$th AECM iteration. Let $\thetab^{(t)} = (\bb^{(t)}, \Lb_{1}^{(t)}, \Lb_{2}^{(t)}, \tau^{2(t)})$ be the parameter estimate at the end of the $t$th AECM iteration. Define $\Lambdab_i^{(t)} = \Zb_{i(1)} \Lb_{(1)}^{(t)} \Lb_{(1)}^{(t) \top} \Zb_{i(1)}^\top + \Ib_{m_i}$ for $i=1, \ldots, n$, where, as defined in Section \ref{sec:mmtr_notation}, $\Lb_{(1)}^{(t)} = \Lb_2^{(t)} \otimes \Lb_1^{(t)}$ and $\Zb_{i(1)}$ is the matrix whose $j$th row is $\zb_{ij(1)}^\top$. The first AECM cycle updates $(\bb, \tau^{2})$ given $(\Lb_{1}^{(t)}, \Lb_{2}^{(t)})$. Marginalizing over $\cb_{i(1)}$ in \eqref{eq:e1} implies that 
\begin{align}
  \label{eq:marg}
  \yb_i =  \Xb_{i} \bb + \bar \eb_{i}, \quad \EE(\overline \eb_{i}) = \zero, \quad \cov(\overline \eb_{i}) = \tau^{2} \Lambdab_{i}^{(t)}, \quad i = 1, \ldots, n.
\end{align}
Let $\Lambdab_i^{{1}/{2}(t)}$ be any matrix such that $\Lambdab_i^{{1}/{2}(t)} \Lambdab_i^{{1}/{2} (t)\top} = \Lambdab_i^{(t)}$, and $\breve \yb^{(t)} \in \RR^{N}$ and $\breve{\Xb}^{{(t)}} \in \RR^{N \times P_{1} P_{2}}$
be the scaled response and design matrix with $\Lambdab_i^{-{1}/{2}(t)} \yb_i \in \RR^{m_{i}}$ and $\Lambdab_i^{-{1}/{2}(t)} \Xb_i$ as their $i$th row blocks. Then, \eqref{eq:marg} implies that the loss for estimating $\bb$ is $\| \breve{\yb}^{{(t)}} - \breve{\Xb}^{{(t)}} \bb \|^{2}$. We jointly estimate ($\bb^{(t+1)}, \tau^{2(t+1)}$) by solving the following scaled lasso optimization problem  \citep{sun2012scaled}:
\begin{align}\label{eq:S_B}
  (\bb^{{(t+1)}}, \tau^{(t+1)}) = \underset{\bb \in \RR^{P_{1}P_{2}}, \tau \in \RR}{\argmin} \frac{\quad \| \breve{\yb}^{{(t)}} - \breve{\Xb}^{{(t)}} \bb \|^{2}}{2N \tau} + \frac{\tau}{2} + \lambda_{\Bb} \| \bb \|_1,
\end{align}
where $\lambda_{\Bb} \geq 0$ is a tuning parameter and $\| \bb \|_1 = \sum_{i=1}^{P_{1}} \sum_{j=1}^{P_{2}} | b_{ij} |$ is the lasso penalty. Given $\lambda_{\Bb}$, we estimate $\bb^{{(t+1)}}$ and $\tau^{2(t+1)}$ using the \texttt{scalreg} package and update $\thetab^{(t)}$ to $\thetab^{(t+1/3)} = (\bb^{(t+1)}, \Lb_{1}^{(t)}, \Lb_{2}^{(t)}, \tau^{2(t+1)})$, ending the first AECM cycle.

The second cycle treats $\cb_{1(2)}, \ldots, \cb_{n(2)}$ in \eqref{eq:4} as missing data and updates $\Lb_{1}^{(t)}$ to $\Lb_{1}^{(t+1)}$ given $ (\bb^{(t+1)}, \Lb_{2}^{(t)}, \tau^{2(t+1)})$. For $ i = 1, \ldots, n$,  \eqref{eq:4} implies that the complete data are $(\yb_i, \cb_{i(2)})$ and $\mub^{(t)}_{i(2)} = \EE(\cb_{i(2)} \mid \yb_i, \thetab^{{(t+1/3)}})$ and $\Gammab^{(t)}_{i(2)} = \EE(\cb_{i(2)} \cb_{i(2)}^\top| \yb_i, \thetab^{{(t+1/3)}})$ are analytically tractable. If $\lb_{1} =\text{vec}(\Lb_{1})$, then the second cycle E-step's objective is
\begin{align}
  \label{eq:sec-cyc-E}
  \Qcal_{(1)}(\lb_{1}) &\propto -(\Hb_{(1)}^{-{1}/{2}(t)} \gb^{(t)}_{(1)} - \Hb_{(1)}^{{1}/{2}(t)\top} \lb_1)^\top (\Hb_{(1)}^{-{1}/{2}(t)} \gb^{(t)}_{(1)} - \Hb_{(1)}^{{1}/{2}(t)\top} \lb_1),\\
  \Hb_{(1)}^{{(t)}} & = \sum_{i=1}^n \sum_{j=1}^{m_i} (\Ib_{S_1} \otimes \Zb_{ij(1)} \Lb^{(t)}_{2}) \Gammab_{i(2)}^{(t)} (\Ib_{S_1} \otimes \Lb^{(t)\top}_{2}  \Zb_{ij(1)}^\top), \nonumber\\
  \gb_{(1)}^{(t)} & = \sum_{i=1}^n \left[ \Ib_{S_1} \otimes \left\{ \sum_{j=1}^{m_i} (y_{ij} - \xb_{ij}^{\top} \bb^{{(t+1)}}) \Zb_{ij(1)}\right\} \Lb_{2}^{(t)} \right] \mub_{i(2)}^{(t)},\nonumber
\end{align}
where $\Hb_{(1)}^{{1}/{2}(t)}$ is any matrix such that $\Hb_{(1)}^{{1}/{2}(t)} \Hb_{(1)}^{{1}/{2}(t)\top} = \Hb_{(1)}^{(t)}$ and $\Hb_{(1)}^{-{1}/{2}(t)}$ is any generalized inverse of $\Hb_{(1)}^{{1}/{2}(t)}$. Using \eqref{eq:sec-cyc-E}, $- \Qcal_{(1)}(\lb_{1})$ is equivalent to a squared loss for estimating $\lb_{1}$.

The conditional maximization (CM) step in the second cycle minimizes $- \Qcal_{(1)}(\lb_{1})$ using the group lasso penalty on the columns of $\Lb_{1}$. Let $\lb_{(1):i}$ is the $i$th $Q_{1}$-dimensional block of $\lb_{1}$, which corresponds to the $i$th column of $\Lb_1$. Then, $\lb_{1}$ is estimated using the \texttt{gglasso} package as
\begin{align}
  \label{eq:sec-cyc-M}
  \lb_{1}^{(t+1)} = \argmin_{\lb_1 \in \RR^{Q_{1} S_{1}}} \; \big \lVert \Hb_{(1)}^{-{1}/{2}(t)} \gb^{(t)}_{(1)} - \Hb_{(1)}^{{1}/{2}(t)\top} \lb_1 \big \rVert_2^{2}
    + \lambda_{\Lb} \sum_{i=1}^{S_1} \lVert \lb_{(1):i} \rVert_2,
\end{align}
where  $\Lb^{(t+1)}_{1} = \text{vec}^{{-1}}(\lb_{1}^{(t+1)})$, $\lambda_{\Lb} \geq 0$ is a tuning parameter, and $\lVert \cdot \rVert_2$ is the Euclidean norm. The second cycle ends by updating $\thetab^{(t+1/3)}$ to $\thetab^{(t+2/3)} = (\bb^{(t+1)}, \Lb_{1}^{(t+1)}, \Lb_{2}^{(t)}, \tau^{2(t+1)})$.

Finally, the third AECM cycle updates $\thetab^{(t+2/3)}$ to $\thetab^{(t+1)}$ by updating $\Lb_{2}$ given $ (\bb^{(t+1)}, \Lb_{1}^{(t+1)}, \tau^{2(t+1)})$. The E and CM steps in this cycle change the indices 1 to 2 and vice versa in the second cycle E and CM steps. Specifically, for $i=1, \ldots, n$, the complete data are $(\yb_i, \cb_{i(1)})$ and $\mub^{(t)}_{i(1)} = \EE(\cb_{i(1)} \mid \yb_i, \thetab^{{(t+2/3)}})$ and $\Gammab^{(t)}_{i(1)} = \EE(\cb_{i(1)} \cb_{i(1)}^\top| \yb_i, \thetab^{{(t+2/3)}})$ are analytically tractable. The E step objective for the third cycle, $ \Qcal_{(2)}(\lb_{2})$, swaps the indices 1 and 2 in $ \Qcal_{(1)}(\lb_{1})$ as defined in \eqref{eq:sec-cyc-E}, where $\lb_{2}=\text{vec}(\Lb_{2})$ and $\Hb_{(2)}^{{(t)}}$ and $ \gb_{(2)}^{(t)}$ are obtained by swapping indices 1 and 2 in $\Hb_{(1)}^{(t)}$ and $\gb_{(1)}^{(t)}$. Following \eqref{eq:sec-cyc-M}, the  CM step in the third cycle is
\begin{align}
  \label{eq:third-cyc-M}
  \lb_{2}^{(t+1)} = \argmin_{\lb_2 \in \RR^{Q_{2} S_{2}}} \; \big \lVert \Hb_{(2)}^{-{1}/{2}(t)} \gb^{(t)}_{(2)} - \Hb_{(2)}^{{1}/{2}(t)\top} \lb_2 \big \rVert_2^{2}
    + \lambda_{\Lb} \sum_{i=1}^{S_2} \lVert \lb_{(2):i} \rVert_2,
\end{align}
where $\Lb^{(t+1)}_{2} = \text{vec}^{{-1}}(\lb_{2}^{(t+1)})$, $\lb_{(2):i}$ is the $i$th $Q_{2}$-dimensional block of $\lb_{2}$,  which corresponds to the $i$th column of $\Lb_2$, and $ \lb_{2}^{(t+1)} $ is estimated using the \texttt{gglasso} package. This cycle finishes the $(t+1)$th iteration of the AECM algorithm, updating the $\thetab^{{(t)}}$ to $\thetab^{(t+3/3)} = \thetab^{(t+1)} = (\bb^{(t+1)}, \Lb_{1}^{(t+1)}, \Lb_{2}^{(t+1)}, \tau^{2(t+1)})$. Algorithm \ref{algo1} summarizes the AECM algorithm.

\begin{algorithm}
\caption{MMTR AECM Algorithm}
\label{algo1}
  \textbf{First cycle:}
  \begin{itemize}
    \item \textbf{Regularized weighted least squares:} Given ${\Lb}_1^{(t)}$ and $\Lb_2^{(t)}$, compute $\Bb^{{(t+1)}}$ and ${\tau}^{2(t+1)}$ by solving \eqref{eq:S_B}. Update $\thetab^{{(t)}}$ to $\thetab^{(t+1/3)} = (\bb^{(t+1)}, \Lb_{1}^{(t)}, \Lb_{2}^{(t)}, \tau^{2(t+1)})$.
  \end{itemize}

  \textbf{Second cycle:}
  \begin{itemize}
    \item  \textbf{E step:} Compute $\Hb_{(1)}^{{(t)}}$ and $\gb_{(1)}^{(t)}$ in \eqref{eq:sec-cyc-E}.
    \item  \textbf{CM step:} Given ${\thetab}^{(t+1/3)}$, compute $\Lb_{1}^{{(t+1)}}$ by solving \eqref{eq:sec-cyc-M}. Update $\thetab^{(t+1/3)}$ to $\thetab^{(t+2/3)} = (\bb^{(t+1)}, \Lb_{1}^{(t+1)}, \Lb_{2}^{(t)}, \tau^{2(t+1)})$.
  \end{itemize}
  \textbf{Third cycle:}
  \begin{itemize}
    \item \textbf{E step:} Compute $\Hb_{(2)}^{{(t)}}$ and $\gb_{(2)}^{(t)}$ by swapping 1 and 2 in \eqref{eq:sec-cyc-E}.
    \item  \textbf{CM step:} Given ${\thetab}^{(t+2/3)}$, compute $\Lb_{2}^{{(t+1)}}$ by solving \eqref{eq:third-cyc-M}. Update $\thetab^{(t+2/3)}$ to $\thetab^{(t+1)} = (\bb^{(t+1)}, \Lb_{1}^{(t+1)}, \Lb_{2}^{(t+1)}, \tau^{2(t+1)})$.
  \end{itemize}
\end{algorithm}

We post-process the $\Lb_1$ and $\Lb_2$ estimates at the end of every iteration
to improve the efficiency and stability of Algorithm \ref{algo1}. For $t=1, 2,
\ldots, \infty$, any zero columns in $\Lb_1^{(t)}$ and $\Lb_2^{(t)}$ are
removed. Then, $\Lb_1^{(t)}$ and $\Lb_2^{(t)}$ are updated following Proposition
\ref{prop2} such that the maximum diagonal values of $\Sigmab_1^{(t)}$ and
$\Sigmab_2^{(t)}$ equal 1 and the ``overall variance'' $\sigma^{2(t)}$ is split equally between $\Sigmab_1^{(t)}$ and $\Sigmab_2^{(t)}$. This update maintains similar ranks for $\Lb_1^{(t)}$ and $\Lb_2^{(t)}$ across different values of tuning parameter $\lambda_{\Lb}$ and random initializations $\Lb_1^{(0)}$ and $\Lb_2^{(0)}$.

The MMTR's parameter estimate $\thetab^{{(\infty)}}$ is a regularized maximum
likehood estimate of $\thetab$. Let $\ell(\thetab)$ be the negative log likehood
function implied by \eqref{eq:px-tensr-mdl}, $\thetab_{1} = (\Bb, \tau^{2})$,
$\thetab_{2} = \Lb_{1}$, $\thetab_{3} = \Lb_{2}$, $\thetab = (\thetab_{1},
\thetab_{2}, \thetab_{3})$, and    $\Pcal_{\lambda_{\Bb}}(\thetab_{1})$,
$\Pcal_{\lambda_{\Lb}}(\thetab_{2})$, and $\Pcal_{\lambda_{\Lb}}(\thetab_{3})$
are the penalties on $\Bb$, $\Lb_1$, and $\Lb_2$ in  \eqref{eq:S_B},
\eqref{eq:sec-cyc-M}, and \eqref{eq:third-cyc-M}, respectively. Then, the
following proposition shows that $\{\thetab^{(t)}\}_{t=1}^{\infty}$ sequence produced by
Algorithm \ref{algo1} converges under weak assumptions.
\begin{proposition}\label{prop-conv}
The MMTR objective is $f(\thetab) = \ell(\thetab) + \Pcal_{\lambda_{\Bb}}(\thetab_{1}) + \Pcal_{\lambda_{\Lb}}(\thetab_{2})+\Pcal_{\lambda_{\Lb}}(\thetab_{3})$. Let $\Mcal(\cdot)$ be the function that maps $\thetab^{(t)}$ to $\thetab^{(t+1)}$ using Algorithm \ref{algo1}. Then, each iteration of Algorithm \ref{algo1} does not increase $f(\thetab)$. Furthermore, assume that the parameter space $\Thetab$ is compact and $f(\thetab) = f\{\Mcal(\thetab)\}$ only for the stationary points of $f(\thetab)$. Then, the $\{\thetab^{(t)} \}_{t=1}^{\infty}$ sequence converges to a stationary point.

\end{proposition}

\section{Experiments}

\subsection{Setup}

MMTR's performance was evaluated using simulated and real data. The
high-dimensional penalized competitors were based on quasi-likelihood
\citep{Lietal22}, trace regression \citep{zhao2017}, posterior mode estimation
\citep{heiling2023glmmpen}, and generalized estimating equations (GEE)
\citep{Zhaetal19}. The posterior mode estimation algorithm implemented in
\texttt{glmmPen} R package failed with an error in the pre-screening step, so we
do not present comparisons with this method. In low-dimensional settings, we
used \texttt{lme4} \citep{Batetal13} for maximum likelihood estimation,
resulting in total four competing methods. The penalized quasi-likelihood (PQL)
method, implemented using the \texttt{scalreg} R package, and  \texttt{lme4}
used MMTR's vectorized form in \eqref{eq:3} for parameter estimation. In
contrast, trace regression and GEE did not have random effects, so we employed
\texttt{TensorReg} \citep{tensorreg_lib}  and \texttt{SparseReg}
\citep{sparsereg_lib} MATLAB packages to fit trace regression and GEE models.
MMTR used Algorithm 1 for parameter estimation.

The marginal model in \eqref{eq:marg} was used for simulations. Specifically,
\begin{align}
  \label{eq:lin_mod}
\yb = \Xb \bb + \eb, \quad \eb \sim N(\zero, \Lambdab), \quad \yb \in \RR^{N}, \; \Xb \in \RR^{{N \times P}}, \; \bb \in \RR^{P}, \; \eb \in \RR^{N},
\end{align}
where $\Lambdab$ was either a diagonal or block diagonal covariance matrix. We
used MMTR and GEE as the true models for simulating the data. The GEE model
specified $\Lambdab$ directly, whereas MMTR  specified $\Lambdab$ through random
effects design matrix, $\Sigmab_1, \Sigmab_2$, and $\tau^2$.

Except for \texttt{lme4}, the performance of MMTR and its competitors depended on model-specific choices. The PQL method replaced $\Lambdab$ with a proxy matrix $c^{2}\Ib$ and used a lasso penalty on $\bb$. Sparse trace regression set $\Lambdab = \tau^{2} \Ib$ and $\Bb = \text{vec}^{-1}(\bb)$ to be a low-rank and sparse matrix. We used two sparse GEE models: one with an unstructured covariance and the other with an equicorrelation covariance structure. Both these models employed a lasso penalty on $\bb$. The unstructured covariance model assumed that $\eb_{i}$ was distributed as $N(\zero, \tilde{\Lambdab})$  for every $i$, implying that $m_{i} = m$ and $\Lambdab = \Ib_{n} \otimes \tilde{\Lambdab}$. Due to this constraint, we used this model only in simulations. The sparse GEE equicorrelation model assumed that $\Lambdab$ was a correlation matrix with off-diagonal elements equal to $\alpha \in (0, 1)$. Unlike the previous GEE model, the $m_{i}$s in this model were allowed to be unequal, enabling its use for simulated and real data analyses. The PQL, trace regression, and GEE tuning parameters were selected using the recommended approach or cross-validation. MMTR had two tuning parameters: $\lambda_{\Bb}$ in \eqref{eq:S_B}, which controlled the sparsity of $\Bb$, and $\lambda_{\Lb}$  in \eqref{eq:sec-cyc-M} and \eqref{eq:third-cyc-M}, which determined the ranks of $\Sigmab_{1}$ and $\Sigmab_{2}$. We varied $\log \lambda_{\Bb}$ and $\log \lambda_{\Lb}$ on a grid and chose the best  $(\lambda_{\Bb}, \lambda_{\Lb})$ pair with the minimum extended Bayesian information criterion \citep{CheChe08}.

The empirical performance in estimation and prediction was quantified using relative estimation error and mean squared prediction error. Let $\hat \xib$ be the estimate of the parameter $\xib$, where $\xib \in \{\Bb, \Lambdab, \Sigmab_{1}. \Sigmab_2\}$, and $\hat \yb_{i}$ be the predicted value of $\yb_{i}$ in the test data. Then, the relative estimation error and mean square prediction errors are
\begin{align}
  \label{eq:err}
  \text{err}_{\xib} = {\| \hat \xib - \xib \|_{F}} / {\| \xib \|_{F}}, \quad \xib \in \{\Bb, \Lambdab, \Sigmab_{1}, \Sigmab_{2}\}, \quad \text{MSPE} =  \sum_{i=1}^{n^{*}} \| \yb_{i} - \hat \yb_{i}\|_{2}^{2} / n^{*},
\end{align}
where $ n^{*}$ is test data sample size and $\| \cdot \|_{F}$ and $\|\cdot \|_2$ are the Frobenius and Euclidean norms.

\subsection{Simulation Using the MMTR Model}

We simulated the data from MMTR in \eqref{eq:px-tensr-mdl}, resulting in a misspecified model for the MMTR's competitors. Two choices of parameter dimensions were used: $\Bb \in \mathbb{R}^{5 \times 5}$, $\Lb_1 \in \mathbb{R}^{5 \times 2}$, $\Lb_2 \in \mathbb{R}^{5 \times 2}$  (\textbf{Case 1}), and (2) $\Bb \in \mathbb{R}^{10 \times 10}$,   $\Lb_1 \in \mathbb{R}^{10 \times 3}$, $\Lb_2 \in \mathbb{R}^{10 \times 3}$  (\textbf{Case 2}). For the two cases, $\lfloor{0.4 P_1 P_2} \rfloor$  entries in $\Bb$ were randomly set to zero, and the remaining $\Bb$ entries were sampled with replacement from a uniform grid in $[-10, 10]$ with a step size of 0.5. The entries in $\Lb_1$ and $\Lb_2$ were sampled from a continuous uniform distribution on $[-1, 1]$, implying that the row and column covariance matrices were $\Sigmab_1 = \Lb_1 \Lb_1^\top$ and $\Sigmab_2 = \Lb_2 \Lb_2^\top$. The covariance matrices $\Sigmab_1$ and $\Sigmab_2$ were identified by scaling their nonzero singular values such that their product equaled one. The $\tau^2$ parameter was fixed at 0.5 for all simulations.

We varied $m_{i}$, $n$, and $N$ to evaluate their impact on MMTR's performance. For Cases 1 and 2, we simulated data with $n \in \{18, 27, 54, 81\}$ and $n \in \{64, 96, 192, 288\}$, respectively, ensuring that the ratio of $n$ to the parameter dimension was 0.4, 0.6, 1.2, and 1.8. For every case and $n$ combination, the number of observations per group was set to $m_{i} \in \{2, 6\}$, resulting in 16 different simulation scenarios. The \texttt{lme4} package required $N \geq n Q_1 Q_2$ for estimating covariance parameters, so it was excluded as a competitor. We performed 100 replications for each scenario and evaluated MMTR's performance using three competitors.

We estimated $\Bb$ and $\Lambdab$ in \eqref{eq:lin_mod} by selecting tuning parameters via ten-fold cross-validation. MMTR's tuning parameters, $\lambda_{\Bb}$ and $\lambda_{\Lb}$, ranged from $10^{-4}$ and $0.04$ and $10^{{-4}}$ and 1, respectively, on an equally spaced 10-by-10 log-scale grid. Sparse Kruskal regression  estimated only $\Bb$ by varying its tuning parameter from $10^{{-2}}$ to $10$ on an equally spaced log-scale grid. The quasi-likelihood method also required $N \geq n Q_1 Q_2$ for estimating $\Lambdab$, so it estimated only $\Bb$. For the two GEE models and the quasi-likelihood method, the tuning parameters for estimating $\Bb$ varied evenly over 10 values from $10^{-4}$ and $10$ on a log-scale grid. This resulted in $\Bb$ estimates for four methods and $\Lambdab$ estimates for only two.

We evaluated the performance of MMTR and its competitors using the metrics in \eqref{eq:err}  (Tables  \ref{table:mmtr_sim1} and \ref{table:mmtr_sim2}). For  $\Bb$ estimation, all four methods showed similar estimation accuracy across all simulation scenarios, demonstrating that mean parameter estimation was robust to model misspecification. In contrast, accurate $\Lambdab$ estimation was more challenging and required $m_{i}$ and $n$ to be sufficiently large. Furthermore, MMTR  outperformed both GEE models across all scenarios in $\Lambdab$ estimation, indicating that covariance estimation was sensitive to model misspecification. MMTR's relative accuracy increased with $n$ and $m_{i}$, whereas both GEE models' performance remained unaffected by increasing $n$ and $m_{i}$.

Unlike its competitors, MMTR estimated $\Sigmab_{1}$ and $\Sigmab_{2}$ using Proposition \ref{prop2} (Tables \ref{table:mmtr_covar1} and \ref{table:mmtr_covar2}). The results showed that $m_{i}$ had a greater impact on estimation accuracy of $\Sigmab_{1}$ and $\Sigmab_{2}$ than $N$. Specifically, for a fixed $N$, MMTR achieved higher accuracy when $m_{i} = 6$ than when $m_{i} = 2$; see the blue and red highlighted rows in Tables \ref{table:mmtr_covar1} and \ref{table:mmtr_covar2}. In summary, while all the methods had similar accuracy in estimating $\Bb$, MMTR outperformed its competitors in accurately estimating $\Sigmab_{1}, \Sigmab_{2}$, and $\Lambdab$, particularly for large $n$ and $m_{i}$.

\begin{table}[H]
  \caption{
    Relative estimation error in Case 1. The simulation follows MMTR in \eqref{eq:px-tensr-mdl}. Every estimation error is averaged over 100 simulation replications, with the Monte Carlo errors in parenthesis. A small relative error indicates high estimation accuracy.  GEE un. and GEE eq. denote the GEE models with unstructured and equicorrelated covariance matrix. Kruskal and Quasi denote the sparse Kruskal and quasi-likelihood methods.}
  \label{table:mmtr_sim1}
  \centering
  \resizebox{\textwidth}{!}{%
    \begin{tabular}[t]{|r|c|c|c|c|c|c|c|c|}
      \hline
      group size &
      \multicolumn{5}{c|}{err$_{\Bb}$} &
      \multicolumn{3}{c|}{err$_{\Lambdab}$} \\
      \hline
      & MMTR & Kruskal & GEE un. & GEE eq. & Quasi & MMTR & GEE un. & GEE eq. \\
      \hline
      \multicolumn{9}{|l|}{\textbf{n = 18}} \\
      \hline
      \hspace{1em} $m_i=2$ & 0.130 (0.0440) & 0.130 (0.0500) & 0.170 (0.0650) & 0.110 (0.0360) & 0.620 (0.2000) & 0.91 (0.052) & 0.95 (0.094) & 0.86 (0.052) \\
      \hline
      \hspace{1em} $m_i=6$ & 0.044 (0.0100) & 0.051 (0.0120) & 0.082 (0.0340) & 0.047 (0.0120) & 0.110 (0.0310) & 0.74 (0.330) & 0.93 (0.027) & 0.91 (0.037) \\
      \hline
      \multicolumn{9}{|l|}{\textbf{n = 27}} \\
      \hline
      \hspace{1em} $m_i=2$ & 0.099 (0.0290) & 0.077 (0.0250) & 0.087 (0.0430) & 0.073 (0.0230) & 0.310 (0.1500) & 0.99 (0.250) & 0.81 (0.160) & 0.86 (0.055) \\
      \hline
      \hspace{1em} $m_i=6$ & 0.028 (0.0076) & 0.037 (0.0100) & 0.043 (0.0120) & 0.035 (0.0097) & 0.083 (0.0220) & 0.49 (0.150) & 1.20 (0.470) & 0.91 (0.036) \\
      \hline
      \multicolumn{9}{|l|}{\textbf{n = 54}} \\
      \hline
      \hspace{1em} $m_i=2$ & 0.051 (0.0110) & 0.048 (0.0110) & 0.046 (0.0120) & 0.045 (0.0120) & 0.110 (0.0280) & 0.92 (0.290) & 0.81 (0.071) & 0.85 (0.056) \\
      \hline
      \hspace{1em} $m_i=6$ & 0.016 (0.0040) & 0.026 (0.0060) & 0.026 (0.0066) & 0.024 (0.0058) & 0.052 (0.0110) & 0.29 (0.090) & 0.88 (0.046) & 0.91 (0.037) \\
      \hline
      \multicolumn{9}{|l|}{\textbf{n = 81}} \\
      \hline
      \hspace{1em} $m_i=2$ & 0.034 (0.0066) & 0.037 (0.0096) & 0.035 (0.0092) & 0.035 (0.0091) & 0.082 (0.0210) & 0.59 (0.220) & 0.83 (0.062) & 0.85 (0.054) \\
      \hline
      \hspace{1em} $m_i=6$ & 0.013 (0.0034) & 0.021 (0.0057) & 0.021 (0.0056) & 0.020 (0.0056) & 0.042 (0.0094) & 0.24 (0.074) & 0.89 (0.043) & 0.91 (0.036) \\
      \hline
    \end{tabular}}
\end{table}

\begin{table}[H]
  \caption{
    Relative estimation error in Case 2. The simulation follows MMTR in \eqref{eq:px-tensr-mdl}. Every estimation error is averaged over 100 simulation replications, with the Monte Carlo errors in parenthesis. A small relative error indicates high estimation accuracy. GEE un. and GEE eq. denote the GEE models with unstructured and equicorrelated covariance matrix. Kruskal and Quasi denote the sparse Kruskal and quasi-likelihood methods.}
  \label{table:mmtr_sim2}
  \centering
  \resizebox{\textwidth}{!}{%
    \begin{tabular}[t]{|r|c|c|c|c|c|c|c|c|}
      \hline
      group size &
      \multicolumn{5}{c|}{err$_{\Bb}$} &
      \multicolumn{3}{c|}{err$_{\Lambdab}$} \\
      \hline
      & MMTR & Kruskal & GEE un. & GEE eq. & Quasi & MMTR & GEE un. & GEE eq. \\
      \hline
      \multicolumn{9}{|l|}{\textbf{n = 64}} \\
      \hline
      \hspace{1em} $m_i=2$ & 0.120 (0.0200) & 0.55 (0.097) & 0.120 (0.0230) & 0.096 (0.0160) & 0.890 (0.0440) & 0.92 (0.019) & 0.98 (0.0045) & 0.90 (0.019) \\
      \hline
      \hspace{1em} $m_i=6$ & 0.040 (0.0066) & 0.36 (0.050) & 0.039 (0.0091) & 0.033 (0.0050) & 0.150 (0.0350) & 0.87 (0.240) & 0.95 (0.0110) & 0.93 (0.015) \\
      \hline
      \multicolumn{9}{|l|}{\textbf{n = 96}} \\
      \hline
      \hspace{1em} $m_i=2$ & 0.085 (0.0130) & 0.42 (0.062) & 0.071 (0.0240) & 0.057 (0.0090) & 0.740 (0.0980) & 1.10 (0.140) & 0.76 (0.0700) & 0.90 (0.018) \\
      \hline
      \hspace{1em} $m_i=6$ & 0.022 (0.0028) & 0.34 (0.044) & 0.028 (0.0081) & 0.025 (0.0037) & 0.084 (0.0130) & 0.46 (0.098) & 0.84 (0.1800) & 0.93 (0.015) \\
      \hline
      \multicolumn{9}{|l|}{\textbf{n = 192}} \\
      \hline
      \hspace{1em} $m_i=2$ & 0.043 (0.0075) & 0.36 (0.049) & 0.033 (0.0057) & 0.033 (0.0059) & 0.150 (0.0280) & 1.10 (0.200) & 0.86 (0.0260) & 0.90 (0.019) \\
      \hline
      \hspace{1em} $m_i=6$ & 0.012 (0.0018) & 0.33 (0.043) & 0.018 (0.0026) & 0.018 (0.0026) & 0.049 (0.0063) & 0.24 (0.041) & 0.88 (0.0240) & 0.93 (0.015) \\
      \hline
      \multicolumn{9}{|l|}{\textbf{n = 288}}\\
      \hline
      \hspace{1em} $m_i=2$ & 0.028 (0.0045) & 0.34 (0.045) & 0.025 (0.0036) & 0.025 (0.0036) & 0.086 (0.0130) & 0.63 (0.130) & 0.88 (0.0230) & 0.90 (0.019) \\
      \hline
      \hspace{1em} $m_i=6$ & 0.010 (0.0015) & 0.32 (0.042) & 0.014 (0.0021) & 0.014 (0.0021) & 0.038 (0.0049) & 0.18 (0.032) & 0.90 (0.0200) & 0.93 (0.015) \\
      \hline
    \end{tabular}}
\end{table}

\begin{table}[H]
  \begin{minipage}{0.45\linewidth}
    \caption{MMTR's relative estimation error for $\Sigmab_{1}$ and $\Sigmab_{2}$ in Case 1. Every estimation error is averaged over 100 simulation replications, with the Monte Carlo errors in parenthesis. The red and blue lines represent simulations with equal $N$ values.}
    \label{table:mmtr_covar1}
    \centering
    \resizebox{0.6 \textwidth}{!}{%
    \begin{tabular}[t]{|r|c|c|}
      \hline
      group size &
      err$_{\Sigmab_1}$&
      err$_{\Sigmab_2}$\\
      \hline
      \multicolumn{3}{|l|}{\textbf{n = 18}} \\
      \hline
      \hspace{1em} $m_i=2$ & 1.10 (0.230) & 1.10 (0.270) \\
      \hline
      \rowcolor{LightBlue2} \hspace{1em} $m_i=6$ & 0.57 (0.260) & 0.54 (0.250) \\
      \hline
      \multicolumn{3}{|l|}{\textbf{n = 27}} \\
      \hline
      \hspace{1em} $m_i=2$ & 1.00 (0.290) & 1.10 (0.300) \\
      \hline
      \rowcolor{LightPink2} \hspace{1em} $m_i=6$ & 0.40 (0.170) & 0.39 (0.150) \\
      \hline
      \multicolumn{3}{|l|}{\textbf{n = 54}} \\
      \hline
      \rowcolor{LightBlue2} \hspace{1em} $m_i=2$ & 0.84 (0.290) & 0.83 (0.340) \\
      \hline
      \hspace{1em} $m_i=6$ & 0.21 (0.100) & 0.21 (0.085) \\
      \hline
      \multicolumn{3}{|l|}{\textbf{n = 81}} \\
      \hline
      \rowcolor{LightPink2} \hspace{1em} $m_i=2$ & 0.64 (0.320) & 0.63 (0.280) \\
      \hline
      \hspace{1em} $m_i=6$ & 0.17 (0.069) & 0.16 (0.066) \\
      \hline
    \end{tabular}}
\end{minipage}
\hspace{1em}
  \begin{minipage}{0.45\linewidth}
    \caption{MMTR's relative estimation error for $\Sigmab_{1}$ and $\Sigmab_{2}$ in Case 2. Every estimation error is averaged over 100 simulation replications, with the Monte Carlo errors in parenthesis. The red and blue lines represent simulations with equal $N$ values.}
    \label{table:mmtr_covar2}
    \centering
    \resizebox{0.6\textwidth}{!}{%
    \begin{tabular}[t]{|r|c|c|c|c|}
      \hline
      group size &
      err$_{\Sigmab_1}$&
      err$_{\Sigmab_2}$\\
      \hline
      \multicolumn{3}{|l|}{\textbf{n = 64}} \\
      \hline
      \hspace{1em} $m_i=2$ & 1.20 (0.200) & 1.20 (0.170) \\
      \hline
      \rowcolor{LightBlue2} \hspace{1em} $m_i=6$ & 0.80 (0.230) & 0.80 (0.210) \\
      \hline
      \multicolumn{3}{|l|}{\textbf{n = 96}} \\
      \hline
      \hspace{1em} $m_i=2$ & 1.20 (0.190) & 1.20 (0.190) \\
      \hline
      \rowcolor{LightPink2} \hspace{1em} $m_i=6$ & 0.47 (0.160) & 0.49 (0.160) \\
      \hline
      \multicolumn{3}{|l|}{\textbf{n = 192}} \\
      \hline
      \rowcolor{LightBlue2} \hspace{1em} $m_i=2$ & 1.10 (0.200) & 1.10 (0.210) \\
      \hline
      \hspace{1em} $m_i=6$ & 0.24 (0.056) & 0.23 (0.048) \\
      \hline
      \multicolumn{3}{|l|}{\textbf{n = 288}} \\
      \hline
      \rowcolor{LightPink2} \hspace{1em} $m_i=2$ & 0.86 (0.280) & 0.87 (0.260) \\
      \hline
      \hspace{1em} $m_i=6$ & 0.17 (0.041) & 0.18 (0.048) \\
      \hline
    \end{tabular}}
  \end{minipage}
\end{table}

\subsection{Simulations Using the Equicorrelation GEE model}
\label{sec:misspec_results}

We simulated the data using the equicorrelation GEE model based on \eqref{eq:lin_mod} to evaluate the impact of misspecification on MMTR's performance. The mean parameter in \eqref{eq:lin_mod} was generated as in the previous simulation, resulting in two cases: $\Bb \in \mathbb{R}^{5 \times 5}$ (\textbf{Case 1}) and $\Bb \in \mathbb{R}^{10 \times 10} (\textbf{Case 2})$. The $n$ choices for these two cases were the same as in the previous simulation to facilitate comparisons. Based on the previous results, we set $m_{i} = 6$ in both cases because MMTR demonstrated satisfactory performance only with sufficiently large  $m_{i}$. This setup resulted in eight different simulation scenarios. The covariance parameter $\Lambdab$ in this simulation was specified by one correlation parameter $\alpha$. For each of the 100 simulation replications across the eight scenarios, $\alpha$ was drawn from a continuous uniform distribution on $[0.2, 0.8]$. The three MMTR competitors (GEE, Kruskal, and quasi-likelihood), tuning parameter selection methods, and error metrics were identical to those used in the previous simulation.

All methods demonstrated similar accuracy in $\Bb$ estimation, but the equicorrelation GEE model outperformed the others in $\Lambdab$ estimation (Tables \ref{table:equicorr_sim1} and \ref{table:equicorr_sim2}). For both cases, $\Bb$ estimation accuracy improved as $n$ increased, providing further evidence that the mean parameter estimation was robust to model misspecification. Although MMTR's covariance parameter was misspecified relative to the simulation model, its  accuracy in estimating $\Lambdab$ improved up to a point with increasing $n$.
Furthermore, MMTR's covariance estimation errors were significantly smaller for
lower $\alpha$ values than higher values. A larger $\alpha$ corresponds to a
heavily misspecified covariance model, making it more challenging to approximate
using a Kronecker product structure; see supplementary materials for additional
simulation results.  In summary, while model misspecification implied that a larger $n$ and smaller $\alpha$ were required for MMTR's accurate covariance estimation, MMTR showed excellent performance in estimating the mean parameter across all $n$ choices.

\begin{table}[H]
  \caption{Relative estimation error in Case 1 with a misspecified equicorrelated covariance structure in \eqref{eq:lin_mod}. Every estimation error is averaged over 100 simulation replications, with the Monte Carlo errors in parenthesis. A small relative error indicates high estimation accuracy. GEE un.\ and GEE eq.\ denote the GEE models with unstructured and equicorrelated covariance matrix. Kruskal and Quasi denote the sparse Kruskal and quasi-likelihood methods.}
  \label{table:equicorr_sim1}
  \centering
  \resizebox{\textwidth}{!}{%
    \begin{tabular}[t]{|r|c|c|c|c|c|c|c|c|}
      \hline
      group size &
      \multicolumn{5}{c|}{err$_{\Bb}$} &
      \multicolumn{3}{c|}{err$_{\Lambdab}$} \\
      \hline
      & MMTR & Kruskal & GEE un. & GEE eq. & Quasi & MMTR & GEE un. & GEE eq. \\
      \hline
      \hspace{1em} $n=18, m_i=6$ & 0.038 (0.0110) & 0.025 (0.0071) & 0.0480 (0.0270) & 0.0200 (0.0053) & 0.062 (0.0160) & 1.90 (0.75) & 1.20 (2.200) & 0.150 (0.110) \\
      \hline
      \hspace{1em} $n=27, m_i=6$ & 0.022 (0.0065) & 0.020 (0.0046) & 0.0220 (0.0061) & 0.0160 (0.0038) & 0.046 (0.0100) & 1.10 (0.39) & 5.20 (3.000) & 0.110 (0.074) \\
      \hline
      \hspace{1em} $n=54, m_i=6$ & 0.013 (0.0031) & 0.013 (0.0031) & 0.0120 (0.0030) & 0.0110 (0.0026) & 0.029 (0.0057) & 0.76 (0.12) & 0.68 (0.130) & 0.088 (0.066) \\
      \hline
      \hspace{1em} $n=81, m_i=6$ & 0.011 (0.0029) & 0.011 (0.0024) & 0.0095 (0.0022) & 0.0088 (0.0020) & 0.023 (0.0048) & 0.73 (0.11) & 0.40 (0.085) & 0.071 (0.056) \\
      \hline
    \end{tabular}}
\end{table}

\begin{table}[H]
  \caption{Relative estimation error in Case 2 with a misspecified equicorrelated covariance structure in \eqref{eq:lin_mod}. Every estimation error is averaged over 100 simulation replications, with the Monte Carlo errors in parenthesis. A small relative error indicates high estimation accuracy. GEE un.\ and GEE eq.\ denote the GEE models with unstructured and equicorrelated covariance matrix. Kruskal and Quasi denote the sparse Kruskal and quasi-likelihood methods.}
  \label{table:equicorr_sim2}
  \centering
  \resizebox{\textwidth}{!}{%
    \begin{tabular}[t]{|r|c|c|c|c|c|c|c|c|}
      \hline
      group size &
      \multicolumn{5}{c|}{err$_{\Bb}$} &
      \multicolumn{3}{c|}{err$_{\Lambdab}$} \\
      \hline
      & MMTR & Kruskal & GEE un. & GEE eq. & Quasi & MMTR & GEE un. & GEE eq. \\
      \hline
      \hspace{1em} $n=64, m_i=6$ & 0.0180 (0.0023) & 0.36 (0.042) & 0.0140 (0.00490) & 0.0100 (0.00130) & 0.060 (0.0170) & 1.20 (0.30) & 0.42 (0.090) & 0.130 (0.077) \\
      \hline
      \hspace{1em} $n=96, m_i=6$ & 0.0110 (0.0015) & 0.34 (0.039) & 0.0094 (0.00130) & 0.0081 (0.00089) & 0.033 (0.0039) & 0.90 (0.13) & 2.50 (2.800) & 0.098 (0.054) \\
      \hline
      \hspace{1em} $n=192, m_i=6$ & 0.0079 (0.0016) & 0.33 (0.037) & 0.0059 (0.00070) & 0.0057 (0.00066) & 0.019 (0.0016) & 0.72 (0.12) & 0.77 (0.082) & 0.054 (0.033) \\
      \hline
      \hspace{1em} $n=288, m_i=6$ & 0.0066 (0.0015) & 0.32 (0.036) & 0.0047 (0.00053) & 0.0046 (0.00046) & 0.014 (0.0013) & 0.70 (0.13) & 0.42 (0.045) & 0.043 (0.031) \\
      \hline
    \end{tabular}}
\end{table}

\subsection{Real Data: Labeled Faces in the Wild}

We used the Labeled Faces in the Wild (LFW) database for our real data analysis, a benchmark dataset widely used to evaluate the empirical performance of array-variate models \citep{LFWTech,lock2018tensor}. The LFW data contains multiple images of famous individuals and 73 real-valued image attributes for every image \citep{Neeraj2009}. For this analysis, we focused on a subset of 5749 individuals whose images were aligned using deep funneling \citep{Huang2012a}. To ensure sufficiently large $m_i$ values for MMTR, we selected individuals with 4 to 50 images, resulting in a final dataset of 594 individuals and a total of 5124 images.

The data were pre-processed to define the covariates. Each image was converted to grayscale and resized to $32 \times 32$ dimensions using \textit{grayscale} and \textit{resize} functions in the \texttt{imager} library \citep{imager_lib}. This process generated $32 \times 32$ fixed and random effects covariate matrices for all the 5124 images. The response variable was defined as the first principal component of the 73 image attributes for each image.

We randomly split the images into ten training and testing sets. Each test set consisted of four randomly selected images from each of the 115 individuals with at least twelve images, resulting in ten training sets of 4664 images and ten testing sets of 460 images. We fit  MMTR and its four competitors to each training set. Tuning parameter selection methods remain unchanged from the simulations. Unlike the simulations,  \texttt{lme4} and quasi-likelihood methods included random intercepts to account for dependencies within images of the same individual. The GEE model with an unstructured covariance parameter was inapplicable because the $m_i$ values varied across the selected individuals. Their performance was evaluated on the test sets using  MSPE in \eqref{eq:err} and prediction $R^2$.

MMTR achieved the lowest MSPE and highest prediction  $R^2$ among all the methods (Table \ref{table:lfw_mspe}). The averaged estimate of $\Bb$ across the ten training sets revealed key facial features in the images (Figure \ref{fig:mean_bhat}). Based on the simulation results, MSPE, and prediction $R^{2}$, we concluded that MMTR outperforms its competitors in accurately predicting resized LFW images.

\begin{table}[H]
  \caption{MSPE and prediction $R^{2}$. The MSPE and prediction $R^{2}$ values are averaged across ten testing sets of the LFW data, with the Monte Carlo errors in parenthesis. GEE equicorr.\,  Kruskal, and Quasi denote the GEE with equicorrelated covariance matrix, sparse Kruskal, and quasi-likelihood methods.}
  \label{table:lfw_mspe}
  \centering
    \begin{tabular}[t]{|r|c|c|c|c|c|}
      \hline
      & MMTR & Quasi & LME4 & GEE equicorr. & Kruskal \\
      \hline
      MSPE & 2.19 (0.13) & 2.32 (0.14) & 2.81 (0.20) & 3.18 (0.24) & 5.47 (0.42) \\
      \hline
      $R^{2}$ & 0.84 (0.012) & 0.83 (0.013) & 0.79 (0.009) & 0.76 (0.012) & 0.60 (0.018) \\
      \hline
    \end{tabular}
\end{table}

\begin{figure}[H]
\includegraphics[width=\textwidth]{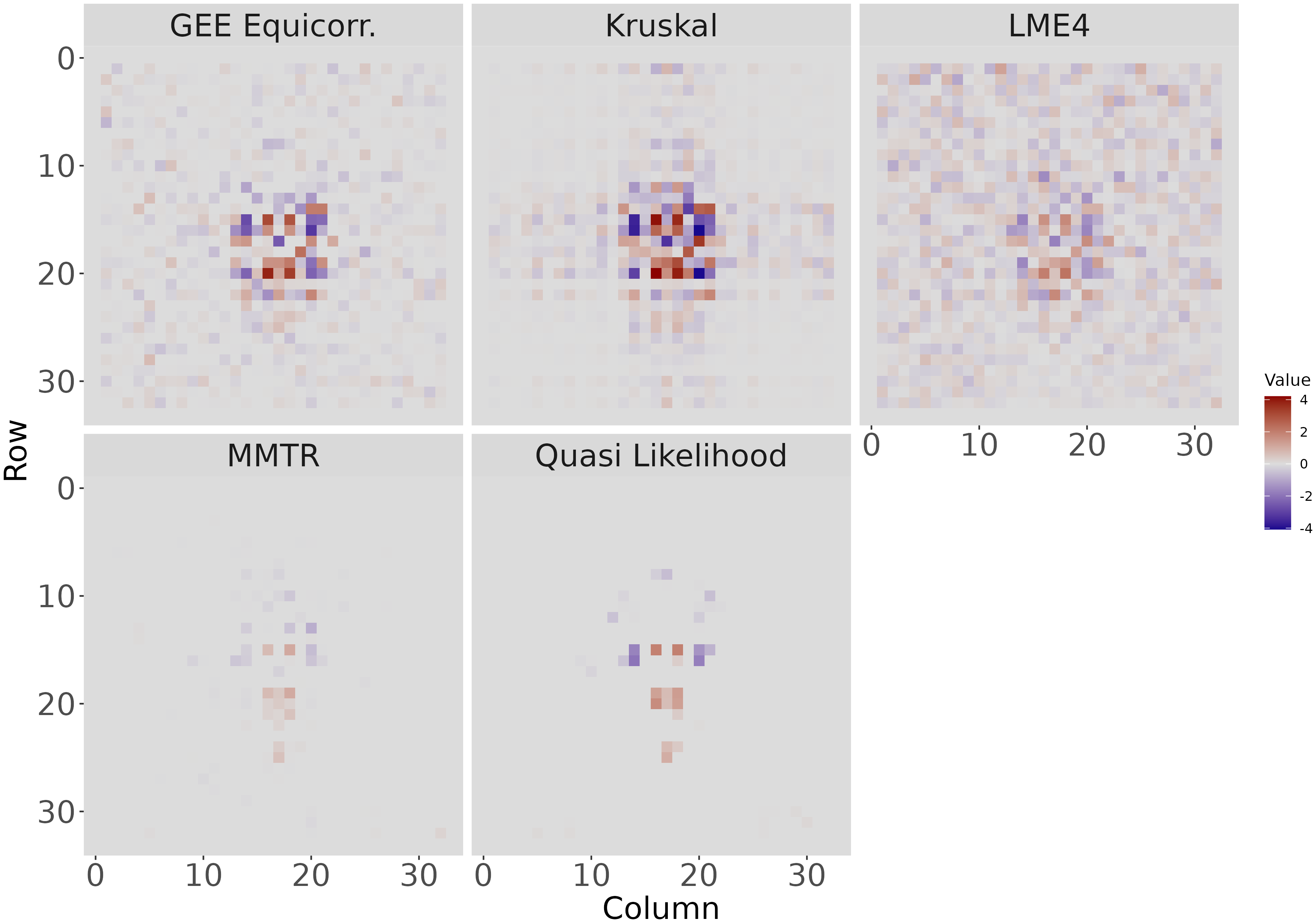}
\caption{Mean parameter ${\Bb}$ estimated by MMTR, quasi-likelihood, lme4, GEE with equicorrelation covariance matrix, and sparse Kruskal methods. The $\Bb$ estimates are averaged across ten training sets of the LFW data.}
\label{fig:mean_bhat}
\end{figure}

\section{Discussion}

We have focused on a correlated scalar responses and matrix-variate covariates,
but MMTR generalizes to array-variate covariates. If the random effect $\Ab_i$
in \eqref{eq:tensr-mdl} is a $Q_1 \times \cdots \times Q_d$ array, then its
separable covariance structure generalizes as $\EE\{\text{vec}(\Ab_i)\} = \zero$
and $\cov \{\text{vec}(\Ab_i)\} = \tau^2 \Sigmab_d \otimes \cdots \otimes
\Sigmab_1$, where $\Ab_i$ is an array-variate normal distribution with zero mean
array and dimension $j$-specific covariance matrix $\Sigmab_j$ ($j=1, \ldots,
d$) \citep{Hof11}. For $j=1, \ldots, d$ and $S_j \ll Q_j$, we assume that there
exists $\Lb_j \in \RR^{Q_j \times S_j}$ such that $\Sigmab_j \approx \Lb_j
\Lb_j^\top$, $\Cb_i$ in \eqref{eq:px-tensr-mdl} is an $S_1 \times \cdots \times
S_d$ array, and $\Lb_j$ is regularized by imposing group lasso penalty on its
columns to enable sparse estimation. Similarly, Propositions \ref{prop1},
\ref{prop2}, and \ref{prop-conv} naturally generalize to the array-variate MMTR
extension.  Finally, the array-variate extension of Algorithm \ref{algo1} has $d+1$ cycles, where the first cycle estimates $\Bb$ and $\tau^2$ using regularized array-variate regression \citep{zhou2013} and the next $d$ cycles estimate $\Lb_j$ given the remaining parameters using a version of \eqref{eq:sec-cyc-M} for $j=1, \ldots, d$. We are developing an MMTR extension for correlated array-variate responses.



The objectives for estimating $\Lb_1$ and $\Lb_2$ in Algorithm \ref{algo1} rely
on the assumption that the random effects and error terms are Gaussian, but our
parameter estimation algorithm generalizes to the case where we only assume the
first and second moments of the random effects exist.  In this case, Algorithm
\ref{algo1} iteratively replaces the first two moments of the random effects,
$\EE(\Ab_i \mid \yb_i, \thetab^{(t)})$ and $\EE \{\text{vec} (\Ab_i) \text{vec}
(\Ab_i)^\top \mid \yb_i, \thetab^{(t)} \}$ ($i=1, \ldots, n$), by those based on
the assumption that  $\Ab_i$ follows a matrix normal distribution, resulting in
the objectives  \eqref{eq:sec-cyc-M} and \eqref{eq:third-cyc-M}.

\section*{Acknowledgments}

Ian Hultman and Sanvesh Srivastava were partially supported by grants from the
National Institutes of Health (1DP2MH126377-01) and the National Science
Foundation (DMS-1854667). The authors thank Kshitij Khare, Joe Lang,
Boxiang Wang, and Dale Zimmerman for their valuable feedback on an earlier
version of this manuscript. The code used in the
experiments is available at \url{https://github.com/IHultman/MMTR}.

\bibliographystyle{Chicago}
\bibliography{papers}

\begin{thebibliography}{}

\bibitem[\protect\citeauthoryear{Barthelme}{Barthelme}{2024}]{imager_lib}
Barthelme, S. (2024).
\newblock {\em imager: Image Processing Library Based on 'CImg'}.
\newblock R package version 1.0.2.

\bibitem[\protect\citeauthoryear{Bates, M{\"a}chler, Bolker, and Walker}{Bates
  et~al.}{2015}]{Bat15}
Bates, D., M.~M{\"a}chler, B.~Bolker, and S.~Walker (2015).
\newblock Fitting linear mixed-effects models using {lme4}.
\newblock {\em Journal of Statistical Software\/}~{\em 67\/}(1), 1--48.

\bibitem[\protect\citeauthoryear{Bates, Maechler, Bolker, and Walker}{Bates
  et~al.}{2013}]{Batetal13}
Bates, D., M.~Maechler, B.~Bolker, and S.~Walker (2013).
\newblock lme4: Linear mixed-effects models using {E}igen and {S4}.
\newblock {\em R package version 1.1-9\/}.

\bibitem[\protect\citeauthoryear{Bradic, Claeskens, and Gueuning}{Bradic
  et~al.}{2020}]{Braetal20}
Bradic, J., G.~Claeskens, and T.~Gueuning (2020).
\newblock Fixed effects testing in high-dimensional linear mixed models.
\newblock {\em Journal of the American Statistical Association\/}~{\em
  115\/}(532), 1835--1850.

\bibitem[\protect\citeauthoryear{Chen and Chen}{Chen and Chen}{2008}]{CheChe08}
Chen, J. and Z.~Chen (2008).
\newblock Extended bayesian information criteria for model selection with large
  model spaces.
\newblock {\em Biometrika\/}~{\em 95\/}(3), 759--771.

\bibitem[\protect\citeauthoryear{Dutilleul}{Dutilleul}{1999}]{Dut99}
Dutilleul, P. (1999).
\newblock The {MLE} algorithm for the matrix normal distribution.
\newblock {\em Journal of Statistical Computation and Simulation\/}~{\em
  64\/}(2), 105--123.

\bibitem[\protect\citeauthoryear{Fan, Gong, and Zhu}{Fan
  et~al.}{2019}]{Fanetal19}
Fan, J., W.~Gong, and Z.~Zhu (2019).
\newblock Generalized high-dimensional trace regression via nuclear norm
  regularization.
\newblock {\em Journal of Econometrics\/}~{\em 212\/}(1), 177--202.

\bibitem[\protect\citeauthoryear{Fan and Li}{Fan and Li}{2012}]{FanLi12}
Fan, Y. and R.~Li (2012).
\newblock Variable selection in linear mixed effects models.
\newblock {\em Annals of statistics\/}~{\em 40\/}(4), 2043.

\bibitem[\protect\citeauthoryear{Fosdick and Hoff}{Fosdick and
  Hoff}{2014}]{FosHof14}
Fosdick, B.~K. and P.~D. Hoff (2014).
\newblock Separable factor analysis with applications to mortality data.
\newblock {\em The annals of applied statistics\/}~{\em 8\/}(1), 120.

\bibitem[\protect\citeauthoryear{Friedman, Hastie, and Tibshirani}{Friedman
  et~al.}{2010}]{FriHasTib10}
Friedman, J., T.~Hastie, and R.~Tibshirani (2010).
\newblock {Regularization paths for generalized linear models via coordinate
  descent}.
\newblock {\em Journal of Statistical Software\/}~{\em 33\/}(1), 1.

\bibitem[\protect\citeauthoryear{Gerard and Hoff}{Gerard and
  Hoff}{2015}]{GerHof15}
Gerard, D. and P.~Hoff (2015).
\newblock Equivariant minimax dominators of the mle in the array normal model.
\newblock {\em Journal of multivariate analysis\/}~{\em 137}, 32--49.

\bibitem[\protect\citeauthoryear{Heiling, Rashid, Li, and Ibrahim}{Heiling
  et~al.}{2023}]{heiling2023glmmpen}
Heiling, H.~M., N.~U. Rashid, Q.~Li, and J.~G. Ibrahim (2023).
\newblock glmmpen: High dimensional penalized generalized linear mixed models.
\newblock {\em The R journal\/}~{\em 15\/}(4), 106.

\bibitem[\protect\citeauthoryear{Heiling, Rashid, Li, Peng, Yeh, and
  Ibrahim}{Heiling et~al.}{2024}]{Hei24}
Heiling, H.~M., N.~U. Rashid, Q.~Li, X.~L. Peng, J.~J. Yeh, and J.~G. Ibrahim
  (2024).
\newblock Efficient computation of high-dimensional penalized generalized
  linear mixed models by latent factor modeling of the random effects.
\newblock {\em Biometrics\/}~{\em 80\/}(1), ujae016.

\bibitem[\protect\citeauthoryear{Hoff}{Hoff}{2011}]{Hof11}
Hoff, P.~D. (2011).
\newblock Separable covariance arrays via the tucker product, with applications
  to multivariate relational data.
\newblock {\em Bayesian Analysis\/}~{\em 6\/}(2), 179--196.

\bibitem[\protect\citeauthoryear{Huang, Mattar, Lee, and Learned-Miller}{Huang
  et~al.}{2012}]{Huang2012a}
Huang, G.~B., M.~Mattar, H.~Lee, and E.~Learned-Miller (2012).
\newblock Learning to align from scratch.
\newblock In {\em NIPS}.

\bibitem[\protect\citeauthoryear{Huang, Ramesh, Berg, and Learned-Miller}{Huang
  et~al.}{2007}]{LFWTech}
Huang, G.~B., M.~Ramesh, T.~Berg, and E.~Learned-Miller (2007, October).
\newblock Labeled faces in the wild: A database for studying face recognition
  in unconstrained environments.
\newblock Technical Report 07-49, University of Massachusetts, Amherst.

\bibitem[\protect\citeauthoryear{Hui, M{\"u}ller, and Welsh}{Hui
  et~al.}{2017}]{Huietal17}
Hui, F.~K., S.~M{\"u}ller, and A.~Welsh (2017).
\newblock Joint selection in mixed models using regularized pql.
\newblock {\em Journal of the American Statistical Association\/}~{\em
  112\/}(519), 1323--1333.

\bibitem[\protect\citeauthoryear{Hui, M{\"u}ller, and Welsh}{Hui
  et~al.}{2021}]{Huietal21}
Hui, F.~K., S.~M{\"u}ller, and A.~H. Welsh (2021).
\newblock Random effects misspecification can have severe consequences for
  random effects inference in linear mixed models.
\newblock {\em International Statistical Review\/}~{\em 89\/}(1), 186--206.

\bibitem[\protect\citeauthoryear{James, Hastie, and Sugar}{James
  et~al.}{2000}]{james2000principal}
James, G.~M., T.~J. Hastie, and C.~A. Sugar (2000).
\newblock Principal component models for sparse functional data.
\newblock {\em Biometrika\/}~{\em 87\/}(3), 587--602.

\bibitem[\protect\citeauthoryear{Kumar, Berg, Belhumeur, and Nayar}{Kumar
  et~al.}{2009}]{Neeraj2009}
Kumar, N., A.~C. Berg, P.~N. Belhumeur, and S.~K. Nayar (2009).
\newblock Attribute and simile classifiers for face verification.
\newblock In {\em 2009 IEEE 12th International Conference on Computer Vision},
  pp.\  365--372.

\bibitem[\protect\citeauthoryear{Li, Cai, and Li}{Li et~al.}{2022}]{Lietal22}
Li, S., T.~T. Cai, and H.~Li (2022).
\newblock Inference for high-dimensional linear mixed-effects models: A
  quasi-likelihood approach.
\newblock {\em Journal of the American Statistical Association\/}~{\em
  117\/}(540), 1835--1846.

\bibitem[\protect\citeauthoryear{Lock}{Lock}{2018}]{lock2018tensor}
Lock, E.~F. (2018).
\newblock Tensor-on-tensor regression.
\newblock {\em Journal of Computational and Graphical Statistics\/}~{\em
  27\/}(3), 638--647.

\bibitem[\protect\citeauthoryear{Lu and Zimmerman}{Lu and
  Zimmerman}{2005}]{LuZim05}
Lu, N. and D.~L. Zimmerman (2005).
\newblock The likelihood ratio test for a separable covariance matrix.
\newblock {\em Statistics \& Probability Letters\/}~{\em 73\/}(4), 449--457.

\bibitem[\protect\citeauthoryear{Ro{\v{c}}kov{\'a} and
  George}{Ro{\v{c}}kov{\'a} and George}{2016}]{RocGeo16}
Ro{\v{c}}kov{\'a}, V. and E.~I. George (2016).
\newblock Fast bayesian factor analysis via automatic rotations to sparsity.
\newblock {\em Journal of the American Statistical Association\/}~{\em
  111\/}(516), 1608--1622.

\bibitem[\protect\citeauthoryear{Seber}{Seber}{2007}]{seber_kronecker}
Seber, G. (2007).
\newblock {\em Special Products and Operators}, Chapter~11, pp.\  233--255.
\newblock John Wiley \& Sons, Ltd.

\bibitem[\protect\citeauthoryear{Slawski, Li, and Hein}{Slawski
  et~al.}{2015}]{slawski2015regularization}
Slawski, M., P.~Li, and M.~Hein (2015).
\newblock Regularization-free estimation in trace regression with symmetric
  positive semidefinite matrices.
\newblock {\em Advances in neural information processing systems\/}~{\em 28}.

\bibitem[\protect\citeauthoryear{Srivastava, von Rosen, and
  Von~Rosen}{Srivastava et~al.}{2008}]{Srietal08}
Srivastava, M.~S., T.~von Rosen, and D.~Von~Rosen (2008).
\newblock Models with a kronecker product covariance structure: estimation and
  testing.
\newblock {\em Mathematical methods of statistics\/}~{\em 17}, 357--370.

\bibitem[\protect\citeauthoryear{Srivastava, Engelhardt, and Dunson}{Srivastava
  et~al.}{2017}]{Sri17}
Srivastava, S., B.~E. Engelhardt, and D.~B. Dunson (2017).
\newblock Expandable factor analysis.
\newblock {\em Biometrika\/}~{\em 104\/}(3), 649--663.

\bibitem[\protect\citeauthoryear{Sun}{Sun}{2019}]{scalreg_lib}
Sun, T. (2019).
\newblock {\em scalreg: Scaled Sparse Linear Regression}.
\newblock R package version 1.0.1.

\bibitem[\protect\citeauthoryear{Sun and Zhang}{Sun and
  Zhang}{2012}]{sun2012scaled}
Sun, T. and C.-H. Zhang (2012).
\newblock Scaled sparse linear regression.
\newblock {\em Biometrika\/}~{\em 99\/}(4), 879--898.

\bibitem[\protect\citeauthoryear{van Dyk}{van Dyk}{2000}]{Dyk00}
van Dyk, D.~A. (2000).
\newblock Fitting mixed-effects models using efficient em-type algorithms.
\newblock {\em Journal of Computational and Graphical Statistics\/}~{\em
  9\/}(1), 78--98.

\bibitem[\protect\citeauthoryear{Van~Loan}{Van~Loan}{2000}]{vanLoan00}
Van~Loan, C.~F. (2000).
\newblock The ubiquitous kronecker product.
\newblock {\em Journal of computational and applied mathematics\/}~{\em
  123\/}(1-2), 85--100.

\bibitem[\protect\citeauthoryear{Verbeke, Molenberghs, and Verbeke}{Verbeke
  et~al.}{1997}]{verbeke1997linear}
Verbeke, G., G.~Molenberghs, and G.~Verbeke (1997).
\newblock {\em Linear mixed models for longitudinal data}.
\newblock Springer.

\bibitem[\protect\citeauthoryear{Yang, Zou, and Bhatnagar}{Yang
  et~al.}{2024}]{gglasso_lib}
Yang, Y., H.~Zou, and S.~Bhatnagar (2024).
\newblock {\em gglasso: Group Lasso Penalized Learning Using a Unified BMD
  Algorithm}.
\newblock R package version 1.5.1.

\bibitem[\protect\citeauthoryear{Yue, Park, Liang, and Shi}{Yue
  et~al.}{2020}]{Yueetal20}
Yue, X., J.~G. Park, Z.~Liang, and J.~Shi (2020).
\newblock Tensor mixed effects model with application to nanomanufacturing
  inspection.
\newblock {\em Technometrics\/}~{\em 62\/}(1), 116--129.

\bibitem[\protect\citeauthoryear{Zhang, Li, Zhou, Zhou, Shen, et~al.}{Zhang
  et~al.}{2019}]{Zhaetal19}
Zhang, X., L.~Li, H.~Zhou, Y.~Zhou, D.~Shen, et~al. (2019).
\newblock Tensor generalized estimating equations for longitudinal imaging
  analysis.
\newblock {\em Statistica Sinica\/}~{\em 29\/}(4), 1977.

\bibitem[\protect\citeauthoryear{Zhang, Shen, and Kong}{Zhang
  et~al.}{2023}]{Zhaetal23}
Zhang, Y., W.~Shen, and D.~Kong (2023).
\newblock Covariance estimation for matrix-valued data.
\newblock {\em Journal of the American Statistical Association\/}~{\em
  118\/}(544), 2620--2631.

\bibitem[\protect\citeauthoryear{Zhao, Niu, and Zhan}{Zhao
  et~al.}{2017}]{zhao2017}
Zhao, J., L.~Niu, and S.~Zhan (2017).
\newblock Trace regression model with simultaneously low rank and row(column)
  sparse parameter.
\newblock {\em Computational Statistics \& Data Analysis\/}~{\em 116}, 1--18.

\bibitem[\protect\citeauthoryear{Zhou}{Zhou}{2017}]{tensorreg_lib}
Zhou, H. (2017).
\newblock {\em Matlab TensorReg Toolbox Version 1.0}.
\newblock \url{https://hua-zhou.github.io/TensorReg}.

\bibitem[\protect\citeauthoryear{Zhou and Gaines}{Zhou and
  Gaines}{2017}]{sparsereg_lib}
Zhou, H. and B.~Gaines (2017).
\newblock {\em Matlab SparseReg Toolbox Version 1.0.0}.
\newblock \url{https://hua-zhou.github.io/SparseReg}.

\bibitem[\protect\citeauthoryear{Zhou and Li}{Zhou and
  Li}{2014}]{zhou2014regularized}
Zhou, H. and L.~Li (2014).
\newblock Regularized matrix regression.
\newblock {\em Journal of the Royal Statistical Society Series B: Statistical
  Methodology\/}~{\em 76\/}(2), 463--483.

\bibitem[\protect\citeauthoryear{Zhou, Li, and Zhu}{Zhou
  et~al.}{2013}]{zhou2013}
Zhou, H., L.~Li, and H.~Zhu (2013).
\newblock Tensor regression with applications in neuroimaging data analysis.
\newblock {\em Journal of the American Statistical Association\/}~{\em
  108\/}(502), 540--552.

\end{thebibliography}

\newpage
\setcounter{equation}{0} 
\setcounter{section}{0}
\setcounter{figure}{0}
\setcounter{table}{0}

\begin{center}
    {\LARGE\bf Supplementary Material: Regularized Parameter Estimation in Mixed
  Model Trace Regression}
\end{center}

\section{Proof of the Propositions}

\subsection{Proof of Proposition 2.1}

Recall that Proposition 2.1 from the main manuscript.
\begin{proposition}\label{prop1}
  Assume that $(\Lb_2)_{j1} = 1$ and $(\Lb_2)_{j2} = \cdots = (\Lb_2)_{jS_2} = 0$. Then, $(\Sigmab_2)_{jj} = 1$ and the $j$th diagonal $Q_1 \times Q_1$ block of $\Sigmab_2 \otimes \Sigmab_1$ is $\Sigmab_1$.
\end{proposition}
\begin{proof}[Proof of Proposition \ref{prop1}] 
Let $\Lb_2 \in \mathbb{R}^{Q_2 \times S_2}$ denote the matrix such that $\Lb_2 \Lb_2^\top = \Sigmab_2$, let $\lb_{i:}$ denote the $i$th row vector of $\Lb_2$ and let $\eb_1 \in \mathbb{R}^{S_2}$ denote the first standard basis vector in $\mathbb{R}^{S_2}$. If the $j$th row of $\Lb_2$ is $\eb_1$, then:
\begin{align*}
  \Lb_2 =
  \begin{bmatrix}
    \lb_{1:}^\top \\
    \vdots \\
    \eb_1^\top \\
    \vdots \\
    \lb_{Q_2:}^\top
  \end{bmatrix} \implies
  \Lb_2 \Lb_2^\top =
  \begin{bmatrix}
    \lb_{1:}^\top \lb_{1:} & \ldots & (\Lb_2)_{11} & \ldots & \lb_{1:}^\top \lb_{Q_2:} \\
    \vdots & \ddots & \vdots && \vdots \\
    (\Lb_2)_{11} & \ldots & 1 & \ldots & (\Lb_2)_{Q_2 1}\\
    \vdots && \vdots & \ddots & \vdots \\
    \lb_{Q_2:}^\top \lb_{1:} & \ldots & (\Lb_2)_{Q_2 1} & \ldots & \lb_{Q_2:}^\top \lb_{Q_2:}
  \end{bmatrix},
\end{align*}
demonstrating that $(\Sigmab_2)_{jj} = 1$ which implies, by definition of the Kronecker product, that the $j$th diagonal $Q_1 \times Q_1$ block of $\Sigmab_2 \otimes \Sigmab_1$ is $\Sigmab_1$.
\end{proof}

\subsection{Proof of Proposition 2.2}

Recall that Proposition 2.2 from the main manuscript.
\begin{proposition}\label{prop2}
  For any matrix $\Lb \in \RR^{Q \times S}$, there exists an orthonormal matrix $\Qb_j \in \RR^{S \times S}$ and constant $c$ such that $c \Lb \Qb_j^\top$ is a matrix whose $j$th row is the first standard basis vector in $\RR^S$.
\end{proposition}

\begin{proof}[Proof of Proposition \ref{prop2}]
Let $\Lb$ be any matrix in $\mathbb{R}^{Q \times S}$ such that $\Lb \Lb^\top = \Sigmab$, and let $\Qb_j \in \RR^{S \times S}$ be the Householder matrix that reflects the $j$th row vector of $\Lb$ onto the first standard basis vector in $\mathbb{R}^S$, denoted as $\eb_1$. Let $\lb_{i:}$ denote the $i$th row vector of $\Lb$, then
\begin{align*}
  \Lb =
  \begin{bmatrix}
    \lb_{1:}^\top \\
    \vdots \\
    \lb_{j:}^\top \\
    \vdots \\
    \lb_{Q:}^\top
  \end{bmatrix} \implies
  \Lb \Qb_j^\top =
  \begin{bmatrix}
    \lb_{1:}^\top \Qb_j^\top \\
    \vdots \\
    \norm{\lb_{j:}}_2 \eb_1^\top \\
    \vdots \\
    \lb_{Q:}^\top \Qb_j^\top
  \end{bmatrix} \implies
  \frac{1}{\norm{\lb_{j:}}_2} \Lb \Qb_j^\top =
  \begin{bmatrix}
    \frac{1}{\norm{\lb_{j:}}_2} \lb_{1:}^\top \Qb_j^\top \\
    \vdots \\
    \eb_1^\top \\
    \vdots \\
    \frac{1}{\norm{\lb_{j:}}_2} \lb_{Q:}^\top \Qb_j^\top
  \end{bmatrix},
\end{align*}
which has the form described in Proposition \ref{prop1}. Because Householder transformations are orthonormal,
\begin{align*}
  \frac{1}{\norm{\lb_{j:}}_2^2} \Lb \Qb_j^\top \Qb_j \Lb^\top = \frac{1}{\norm{\lb_{j:}}_2^2} \Sigmab,
\end{align*}
which is just the initial covariance matrix $\Sigmab$ scaled such that its $j$th diagonal entry is equal to one.
\end{proof}

\subsection{Proof of Proposition 3.1}

Recall Proposition 3.1 from the main manuscript. Let $\ell(\thetab)$ be the negative log likehood function implied by the MMTR, $\thetab_{1} = (\Bb, \tau^{2})$, $\thetab_{2} = \Lb_{1}$, $\thetab_{3} = \Lb_{2}$, $\thetab = (\thetab_{1}, \thetab_{2}, \thetab_{3})$, and  $\Pcal_{\lambda_{\Bb}}(\thetab_{1})$, $\Pcal_{\lambda_{\Lb}}(\thetab_{2})$, and $\Pcal_{\lambda_{\Lb}}(\thetab_{3})$ are the penalties on $\Bb$, $\Lb_1$, and $\Lb_2$, respectively. Then, the Proposition 3.1 from the main manuscript is as follows.
\begin{proposition}\label{prop-conv}
The MMTR objective is $f(\thetab) = \ell(\thetab) +
\Pcal_{\lambda_{\Bb}}(\thetab_{1}) +
\Pcal_{\lambda_{\Lb}}(\thetab_{2})+\Pcal_{\lambda_{\Lb}}(\thetab_{3})$.
Let $\Mcal(\cdot)$ be the function that maps $\thetab^{(t)}$ to
$\thetab^{(t+1)}$ using Algorithm 1 in the main manuscript. Then, each
iteration of Algorithm 1  does not increase $f(\thetab)$. Furthermore,
assume that the parameter space $\Thetab$ is compact and $f(\thetab) =
f\{\Mcal(\thetab)\}$ only for the stationary points of $f(\thetab)$.
Then, the $\{\thetab^{(t)} \}_{t=1}^{\infty}$ sequence converges to a
stationary point.
\end{proposition}

\begin{proof}[Proof of Proposition \ref{prop-conv}]
  Let $\thetab^{(t)} = (\thetab_1^{(t)}, \thetab_2^{(t)}, \thetab_3^{(t)})$ for any $t = 0, 1, 2, \ldots, \infty$. Then, for any non-negative $\lambda_{\Bb}$ and $\lambda_{\Lb}$,
  \begin{align}\label{eq:cyc1-ineq}
    f(\thetab_1^{(t+1)}, \thetab_2^{(t)}, \thetab_3^{(t)}) \leq f(\thetab_1^{(t)}, \thetab_2^{(t)}, \thetab_3^{(t)})
  \end{align}
  because the first cycle implies that $\thetab_1^{(t+1)} =
  \underset{\thetab_1} {\argmin} f(\thetab_1, \thetab_2^{(t)},
  \thetab_3^{(t)})$; see (8) in the main manuscript. We use two
  properties of the AECM algorithm's second cycle. First,
  $-\Qcal_{(1)}(\thetab_2) + \Pcal_{\lambda_{\Lb}}(\thetab_2)$
  majorizes $f(\thetab_1^{(t+1)}, \thetab_2, \thetab_3^{(t)})$ for
  every $\thetab_2$, where $\Qcal_{(1)}$ is defined in (9) in the main
  manuscript; therefore,  $-\Qcal_{(1)}(\thetab_2) +
  \Pcal_{\lambda_{\Lb}}(\thetab_2) \leq f(\thetab_1^{(t+1)},
  \thetab_2, \thetab_3^{(t)})$.  Second, $-\Qcal_{(1)}(\thetab_2^{(t)}) +
  \Pcal_{\lambda_{\Lb}}(\thetab_2^{(t)}) = f(\thetab_1^{(t+1)},
  \thetab_2^{(t)}, \thetab_3^{(t)})$; see (9) and (10) in the main
  manuscript for details. We use these properties to obtain that
  \begin{align}\label{eq:cyc2-ineq}
    f(\thetab_1^{(t+1)}, \thetab_2^{(t+1)}, \thetab_3^{(t)}) &= -\Qcal_{(1)}(\thetab_2^{(t+1)}) + \Pcal_{\lambda_{\Lb}}(\thetab_{2}^{(t+1)}) + \nonumber\\
    & \quad \; f(\thetab_1^{(t+1)}, \thetab_2^{(t+1)},
      \thetab_3^{(t)}) - \{- \Qcal_{(1)}(\thetab_2^{(t+1)}) + \Pcal_{\lambda_{\Lb}}(\thetab_{2}^{(t+1)}) \} \nonumber\\
    &\overset{(i)}{\leq}  - \Qcal_{(1)}(\thetab_2^{(t+1)}) + \Pcal_{\lambda_{\Lb}}(\thetab_{2}^{(t+1)}) \nonumber\\
    &\overset{(ii)}{\leq} -  \Qcal_{(1)}(\thetab_2^{(t)}) + \Pcal_{\lambda_{\Lb}}(\thetab_{2}^{(t)}) \nonumber\\
    &\overset{(iii)}{=} f(\thetab_1^{(t+1)}, \thetab_2^{(t)}, \thetab_3^{(t)}),
  \end{align}
  where $(i)$ follows because $f(\thetab_1^{(t+1)}, \thetab_2^{(t+1)}, \thetab_3^{(t)}) - \{-\Qcal_{(1)}(\thetab_2^{(t+1)}) + \Pcal_{\lambda_{\Lb}}(\thetab_{2}^{(t+1)}) \geq 0$ using the first property of  AECM algorithm's second cycle, $(ii)$ follows because $\thetab_2^{(t+1)} = \underset{\thetab_2} {\argmin}  -\Qcal_{(1)}(\thetab_2) + \Pcal_{\lambda_{\Lb}}(\thetab_{2}) $, and $(iii)$ follows from the second property of  AECM algorithm's second cycle, which implies that $\Qcal_{(1)}(\thetab_2^{(t)}) + \Pcal_{\lambda_{\Lb}}(\thetab_2^{(t)}) = f(\thetab_1^{(t+1)}, \thetab_2^{(t)}, \thetab_3^{(t)})$.

We follow a similar sequence of arguments after swapping the subscripts $1$ and $2$ to obtain lower bound for $f(\thetab_1^{(t+1)}, \thetab_2^{(t+1)}, \thetab_3^{(t)})$ in \eqref{eq:cyc2-ineq} at the end for the AECM algorithm's third cycle. The objective at the end of the third cycle satisfies
  \begin{align}\label{eq:cyc3-ineq}
    f(\thetab_1^{(t+1)}, \thetab_2^{(t+1)}, \thetab_3^{(t+1)})
    &= -\Qcal_{(2)}(\thetab_3^{(t+1)}) + \Pcal_{\lambda_{\Lb}}(\thetab_{3}^{(t+1)}) + \nonumber\\
    & \quad \; f(\thetab_1^{(t+1)}, \thetab_2^{(t+1)}, \thetab_3^{(t+1)}) - \{-\Qcal_{(2)}(\thetab_3^{(t+1)}) + \Pcal_{\lambda_{\Lb}}(\thetab_{3}^{(t+1)}) \} \nonumber\\
    &\overset{(i)}{\leq}   -\Qcal_{(2)}(\thetab_3^{(t+1)}) + \Pcal_{\lambda_{\Lb}}(\thetab_{3}^{(t+1)}) \nonumber\\
    &\overset{(ii)}{\leq}  - \Qcal_{(2)}(\thetab_3^{(t)}) + \Pcal_{\lambda_{\Lb}}(\thetab_{3}^{(t)}) \nonumber\\
    &\overset{(iii)}{=} f(\thetab_1^{(t+1)}, \thetab_2^{(t+1)}, \thetab_3^{(t)}),
  \end{align}
where $(i), (ii),$ and $(iii)$ follow from the same arguments used in  \eqref{eq:cyc2-ineq}, except we swap $1$ and $2$. Finally, \eqref{eq:cyc1-ineq}, \eqref{eq:cyc2-ineq}, and \eqref{eq:cyc3-ineq} imply that
  \begin{align}
      f(\thetab_1^{(t+1)}, \thetab_2^{(t+1)}, \thetab_3^{(t+1)}) \leq f(\thetab_1^{(t+1)}, \thetab_2^{(t+1)}, \thetab_3^{(t)}) \leq
      f(\thetab_1^{(t+1)}, \thetab_2^{(t)}, \thetab_3^{(t)}) \leq f(\thetab_1^{(t)}, \thetab_2^{(t)}, \thetab_3^{(t)}),
  \end{align}
showing that the AECM algorithm does not increase the objective in
every iteration. The objective function $f(\thetab)$ is
bounded for every $\thetab \in \Thetab$ because  $\Thetab$ is compact;
therefore, $f(\thetab^{(t)})$ is a bounded non-increasing sequence, so
$f(\thetab^{(t)})$  converges to $f(\thetab^{(\infty)})$. Our
assumption implies that convergence happens only for stationary
points, so  $\thetab^{(\infty)}$ is a stationary point of
$\thetab^{(t)}$ sequence. The proposition is proved.
\end{proof}

\section{Derivation of the AECM Algorithm Updates}
\label{aecm-updates}

\subsection{Derivation of $\Bb^{(t+1)}$ and $\tau^{2(t+1)}$}

To estimate $\Bb^{(t+1)}$ and $\tau^{2(t+1)}$, we directly use the observed data
log likelihood, $\log f_{\yb | \thetab} (\yb)$, while conditioning on
$\Lb_1^{(t)}$ and $\Lb_2^{(t)}$ which amounts to solving a weighted least
squares; see Equation (7) in the main manuscript. Given $\Lb_1^{(t)}$ and $\Lb_2^{(t)}$,
the log likelhood as a function of $(\Bb, \tau^2)$ is
\begin{align*}
  \log f_{\yb | \thetab} (\yb) & = \sum_{i=1}^n \log f_{\yb_i | \thetab} (\yb_i) \\
  & = -\frac{N}{2}\log 2\pi - \frac{1}{2} \sum_{i=1}^n \text{logdet}(\Sigmab_{\yb_i}^{(t)}) - \frac{1}{2} \sum_{i=1}^n (\yb_i - \Xb_i \bb)^\top \Sigmab_{\yb_i}^{-1(t)} (\yb_i - \Xb_i \bb),
\end{align*}
where $\log f_{\yb_i | \thetab} (\yb_i)$ is the log likelihood sample $i$,
$\Sigmab_{\yb_i}^{(t)} = \tau^2 (\Zb_{i(1)} \Lb_{(1)}^{(t)}
\Lb_{(1)}^{(t)\top} \Zb_{i(1)}^\top + \Ib_{m_i})$, $\Zb_{i(1)}$ is the
matrix whose $j$th row is $\text{vec}(\Zb_{ij})^\top$, $\bb = \text{vec}(\Bb)$, and
$\Lb_{(1)}^{(t)} = \Lb_2^{(t)} \otimes \Lb_1^{(t)}$;  see Section 2.1 in the
main manuscript about the notation. Let $\Lambdab_i^{(t)} = \Zb_{i(1)}
\Lb_{(1)}^{(t)} \Lb_{(1)}^{(t)\top} \Zb_{i(1)}^\top + \Ib_{m_i}$,
$\Lambdab_i^{1/2(t)}$ be any matrix such that $\Lambdab_i^{1/2(t)}
\Lambdab_i^{1/2(t)\top} = \Lambdab_i^{(t)}$, and the scaled response and
predictor matrices be
    \begin{equation*}
      \breve{\yb} =
      \begin{bmatrix}
        \Lambdab_1^{-1/2(t)} \yb_1 \vspace{0.4mm} \\
        \Lambdab_2^{-1/2(t)} \yb_2 \vspace{0.4mm} \\
        \vdots \vspace{0.4mm} \\
        \Lambdab_n^{-1/2(t)} \yb_n
      \end{bmatrix}, \quad
      \breve{\Xb} =
      \begin{bmatrix}
        \Lambdab_1^{-1/2(t)} \Xb_1 \vspace{0.4mm} \\
        \Lambdab_2^{-1/2(t)} \Xb_2 \vspace{0.4mm} \\
        \vdots \vspace{0.4mm} \\
        \Lambdab_n^{-1/2(t)} \Xb_n
      \end{bmatrix}.
    \end{equation*}
Then, maximizing $\log f_{\yb | \thetab} (\yb)$ with respect to  $\bb$
while conditioning on $\Lb_1^{(t)}$ and $\Lb_2^{(t)}$ amounts to
minimizing the quadratic form $(\breve{\yb} - \breve{\Xb} \bb)^\top
(\breve{\yb} - \breve{\Xb} \bb)$ with respect to $\bb$. We estimate
$\Bb^{(t+1)}$ and $\tau^{2(t+1)}$ simultaneously by adding the scaled lasso
penalty terms to this quadratic form as in Equation (8) in the main manuscript
to enforce sparsity in $\Bb^{(t+1)}$ and reduce the bias in estimating
$\tau^{2(t+1)}$. We also update $\thetab^{(t)}$ to $\thetab^{(t+1/3)} =
(\Bb^{(t+1)}, \Lb_1^{(t)}, \Lb_2^{(t)}, \tau^{2(t+1)})$.

\subsection{Derivation of $\Lb_1^{(t+1)}$ and $\Lb_2^{(t+1)}$}

The AECM algorithm treats the matrices $\Cb_i$ $(i = 1, \ldots, n)$ as
missing data and considers the joint distributions of the observed and
missing data $(\yb_i, \cb_{i(k)})$, for $i = 1, \ldots, n$ and $k = 1,
2$; see Equations (9), (10), and (11) in the main manuscript. The AECM algorithm
then maximizes $\EE [\log f_{\yb | \cb, \thetab} (\yb) | \yb, \thetab^{(t+1/3)}]$ with respect to
$\Lb_1$, where $\log f_{\yb | \cb, \thetab} (\yb) | \yb, \thetab^{(t+1/3)}$ is
the complete data log likelihood in Equation of (6) of the main manuscript; see
also Section 3 and Equation (8) in the main manuscript.

Consider the estimation of $\Lb_1$ given $\thetab^{(t+1/3)}$ in the second cycle
of the AECM algorithm. The E step in the second cycle computes the conditional
expectation $\EE [\log f_{\yb | \cb, \thetab} (\yb)
| \yb, \thetab^{(t+1/3)}]$. The CM step  in the second cycle maximizes this objective with respect
to $\Lb_1$. This is equivalent to maximizing $\EE [\log f_{\tilde{\yb} | \cb, \thetab}
(\tilde{\yb}) | \tilde{\yb}, \thetab^{(t+1/3)}]$ with respect to
$\Lb_1$, where $\tilde{\yb} = \yb - \Xb \bb^{(t+1)}$.
The complete data log likelihood of $(\tilde{\yb}, \cb)$ is
\begin{align}\label{eq:supp-llk2}
 \log f_{\tilde{\yb} | \cb, \thetab} (\tilde{\yb}) = -\frac{N}{2} \log 2 \pi \tau^2 - \frac{1}{2 \tau^2}
    \sum_{i=1}^n (\tilde{\yb}_i - \Zb_{i(2)} \Lb_{(2)} \cb_{i(2)})^\top (\tilde{\yb}_i - \Zb_{i(2)} \Lb_{(2)} \cb_{i(2)}),
\end{align}
where, as defined in Section 2.1 of the main manuscript, $\Zb_{i(2)}$
is the matrix whose $j$th row is $\zb_{ij(2)} =
\text{vec}(\Zb_{ij(2)}) = \text{vec}(\Zb_{ij}^\top), \Lb_{(2)} = \Lb_1
\otimes \Lb_2$ and $\cb_{i(2)} = \text{vec}(\Cb_i^\top)$.

We derive the loss function in the E step of the second cycle. For $i = 1 \ldots n$, the joint
distribution of $(\tilde{\yb}_i, \cb_{i(2)})$  is
 \begin{equation*}
   \begin{bmatrix}
     \tilde{\yb}_i \vspace{1mm} \\
     \cb_{i(2)}
   \end{bmatrix} \sim N_{m_i + S_1 S_2} \left(
   \begin{bmatrix}
     \mathbf{0} \\
     \mathbf{0}
   \end{bmatrix}, \tau^2
   \begin{bmatrix}
     \Zb_{i(2)} \Lb_{(2)} \Lb_{(2)}^\top \Zb_{i(2)}^\top + \Ib_{m_i} &
     \Zb_{i(2)} \Lb_{(2)} \vspace{1mm} \\
     \Lb_{(2)}^\top \Zb_{i(2)}^\top &
     \Ib_{S_1 S_2}
    \end{bmatrix} \right).
 \end{equation*}
Using the analytic form of the conditional distribution of $\cb_{i(2)}$ given
$\tilde \yb_i$ in  \eqref{eq:supp-llk2}, gives the loss function $-\Qcal_{(1)}$
for estimating $\Lb_1$. Specifically, the negative of the conditional
expectation of the term in \eqref{eq:supp-llk2} that depends on $\Lb_{1}$ is
\begin{align*}
 - \Qcal_{(1)}(\Lb_1) = \EE \left[ \sum_{i=1}^n (\cb_{i(2)}^\top \Lb_{(2)}^\top \Zb_{i(2)}^\top \Zb_{i(2)} \Lb_{(2)} \cb_{i(2)} -
    2 \tilde{\yb}_i^\top \Zb_{i(2)} \Lb_{(2)} \cb_{i(2)}) | \tilde{\yb}, \thetab^{(t+1/3)} \right],
\end{align*}
and maximizing $\EE[\log f_{\tilde{\yb} | \cb, \thetab} (\tilde{\yb}) | \tilde{\yb}, \thetab^{(t+1/3)}]$ with respect to  $\Lb_1$ amounts to minimizing $-\Qcal_{(1)}(\Lb_1)$ with respect to  $\Lb_1$. Let
\begin{center}
  \begin{tabular}{rcrcl}
    $\mub_{i(2)}^{(t)}$ &
    $=$ &
    $\EE(\cb_{i(2)} | \tilde{\yb}_i, \thetab^{(t+1/3)})$ &
    $=$ &
    $\Lb_{(2)}^{(t)\top} \Zb_{i(2)}^\top (\Zb_{i(2)} \Lb_{(2)}^{(t)} \Lb_{(2)}^{(t)\top} \Zb_{i(2)}^\top + \Ib_{m_i})^{-1} \tilde{\yb}_{i}$, \vspace{1mm} \\
    $\Sigmab_{i(2)}^{(t)}$ &
    $=$ &
    $\text{Cov}(\cb_{i(2)} | \tilde{\yb}_i, \thetab^{(t+1/3)})$ &
    $=$ &
    $\tau^{2(t+1)} (\Ib_{S_1 S_2} + \Lb_{(2)}^{(t)\top} \Zb_{i(2)}^\top \Zb_{i(2)} \Lb_{(2)}^{(t)})^{-1}$, \vspace{1mm} \\
    $\Gammab_{i(2)}^{(t)}$ &
    $=$ &
    $\EE(\cb_{i(2)} \cb_{i(2)}^\top| \tilde{\yb}_i, \thetab^{(t+1/3)})$ &
    $=$ &
    $\Sigmab_{i(2)}^{(t)} + \mub_{i(2)}^{(t)} \mub_{i(2)}^{(t)\top}$.
  \end{tabular}
\end{center}
Then, we can reexpress $-\Qcal_{(1)}(\Lb_1)$ as
\begin{flalign*}
 -\Qcal_{(1)}(\Lb_1) &=
  \EE \left \{ \sum_{i=1}^n \left[ \text{tr}(\Lb_{(2)}^\top \Zb_{i(2)}^\top \Zb_{i(2)} \Lb_{(2)} \cb_{i(2)} \cb_{i(2)}^\top) - 2 \tilde{\yb}_i^\top \Zb_{i(2)} \Lb_{(2)} \cb_{i(2)} \right] | \tilde{\yb}, \thetab^{(t+1/3)} \right \} \\
  &=
  \sum_{i=1}^n \left[ \text{tr}(\Lb_{(2)}^\top \Zb_{i(2)}^\top \Zb_{i(2)} \Lb_{(2)} \Gammab_{i(2)}^{(t)}) - 2 \tilde{\yb}_i^\top \Zb_{i(2)} \Lb_{(2)} \mub_{i(2)}^{(t)} \right] \\
  &=
  \sum_{i=1}^n \left[ \text{tr}(\Zb_{i(2)} \Lb_{(2)} \Gammab_{i(2)}^{(t)} \Lb_{(2)}^\top \Zb_{i(2)}^\top) - 2 \mub_{i(2)}^{(t)\top} \Lb_{(2)}^\top \Zb_{i(2)}^\top \tilde{\yb}_i \right].
  \refstepcounter{equation}\tag{\theequation} \label{eq:Q1}
\end{flalign*}
We simplify this expression by reexpressing the first component of this sum,
$\text{tr}(\Zb_{i(2)} \Lb_{(2)} \Gammab_{i(2)}^{(t)} \Lb_{(2)}^\top
\Zb_{i(2)}^\top)$. If $\lb_1 = \text{vec}(\Lb_1)$, then
\begin{align*}
  \Lb_{(2)}^\top \Zb_{i(2)}^\top & = \Lb_{(2)}^\top
    \begin{bmatrix}
      \zb_{i1(2)} & \zb_{i2(2)} & \ldots & \zb_{im_i(2)}
    \end{bmatrix} \\
  & =
    \begin{bmatrix}
      (\Lb_1^\top \otimes \Lb_2^\top) \zb_{i1(2)} & (\Lb_1^\top \otimes \Lb_2^\top) \zb_{i2(2)} & \ldots & (\Lb_1^\top \otimes \Lb_2^\top) \zb_{im_i(2)}
    \end{bmatrix} \\
  & =
    \begin{bmatrix}
      \text{vec}(\Lb_2^\top \Zb_{i1(2)} \Lb_1) & \text{vec}(\Lb_2^\top \Zb_{i2(2)} \Lb_1) & \ldots & \text{vec}(\Lb_2^\top \Zb_{im_i(2)} \Lb_1)
    \end{bmatrix} \\
  & =
    \begin{bmatrix}
      \text{vec}(\Lb_2^\top \Zb_{i1(1)}^\top \Lb_1) & \text{vec}(\Lb_2^\top \Zb_{i2(1)}^\top \Lb_1) & \ldots & \text{vec}(\Lb_2^\top \Zb_{im_i(1)}^\top \Lb_1)
    \end{bmatrix},\\
  \text{tr}(\Zb_{i(2)} \Lb_{(2)} \Gammab_{i(2)}^{(t)} \Lb_{(2)}^\top \Zb_{i(2)}^\top) & =
    \sum_{j=1}^{m_i} \text{vec}(\Lb_2^\top \Zb_{ij(1)}^\top \Lb_1)^\top \Gammab_{i(2)}^{(t)} \text{vec}(\Lb_2^\top \Zb_{ij(1)}^\top \Lb_1) \\
    & = \sum_{j=1}^{m_i} \text{vec}(\Lb_1)^\top (\Ib_{S_1} \otimes \Zb_{ij(1)} \Lb_{2}) \Gammab_{i(2)}^{(t)} (\Ib_{S_1} \otimes \Lb_{2}^\top \Zb_{ij(1)}^\top) \text{vec}(\Lb_1) \\
    & = \lb_1^\top \left [ \sum_{j=1}^{m_i} (\Ib_{S_1} \otimes \Zb_{ij(1)} \Lb_{2}) \Gammab_{i(2)}^{(t)} (\Ib_{S_1} \otimes \Lb_{2}^\top \Zb_{ij(1)}^\top) \right ] \lb_1.
\end{align*}
Next, we pull $\Lb_1$ out of the second component of \eqref{eq:Q1},
$\mub_{i(2)}^{(t)\top} \Lb_{(2)}^\top \Zb_{i(2)}^\top \tilde{\yb}_i$
to obtain that
\begin{align*}
  \Zb_{i(2)}^\top \tilde{\yb}_i & = \sum_{j=1}^{m_i} \tilde{y}_{ij} \zb_{ij(2)}, \\
  \Lb_{(2)}^\top \Zb_{i(2)}^\top \tilde{\yb}_i & = (\Lb_1^\top \otimes \Lb_2^\top) \sum_{j=1}^{m_i} \tilde{y}_{ij} \zb_{ij(2)} \\
  & = \text{vec} \left[ \Lb_2^\top (\sum_{j=1}^{m_i} \tilde{y}_{ij} \Zb_{ij(2)}) \Lb_1 \right] \\
  & = \text{vec} \left[ \Lb_2^\top (\sum_{j=1}^{m_i} \tilde{y}_{ij} \Zb_{ij(1)}^\top) \Lb_1 \right] \\
  & = \left[ \Ib_{S_1} \otimes \Lb_{2}^\top (\sum_{j=1}^{m_i} \tilde{y}_{ij} \Zb_{ij(1)}^\top) \right] \text{vec}(\Lb_1), \\
  \mub_{i(2)}^{(t)\top} \Lb_{(2)}^\top \Zb_{i(2)}^\top \tilde{\yb}_i & =
    \mub_{i(2)}^{(t)\top} \left[ \Ib_{S_1} \otimes \Lb_{2}^\top (\sum_{j=1}^{m_i} \tilde{y}_{ij} \Zb_{ij(1)}^\top) \right] \lb_1.
\end{align*}
When estimating $\Lb_1$, we condition on $\Lb_{2}^{(t)}$, thus we replace $\Lb_{2}$ with $\Lb_{2}^{(t)}$ in our expression for $-\Qcal_{(1)}(\Lb_1)$. Let
\begin{align*}
  \Hb_{(1)}^{(t)} & = \sum_{i=1}^n \sum_{j=1}^{m_i} (\Ib_{S_1} \otimes \Zb_{ij(1)} \Lb_{2}^{(t)}) \Gammab_{i(2)}^{(t)} (\Ib_{S_1} \otimes \Lb_2^{(t)\top} \Zb_{ij(1)}^\top), \\
  \gb_{(1)}^{(t)} & = \sum_{i=1}^n \left[ \Ib_{S_1} \otimes (\sum_{j=1}^{m_i} \tilde{y}_{ij} \Zb_{ij(1)}) \Lb_{2}^{(t)} \right] \mub_{i(2)}^{(t)},
\end{align*}
then $-\Qcal_{(1)}(\Lb_1) = \lb_1^\top \Hb_{(1)}^{(t)} \lb_1 - 2
\gb_{(1)}^{(t)\top} \lb_1$, which is expressed as the quadratic form
\begin{align}
  \label{eq:L1_quad_form}
-\Qcal_{(1)}(\Lb_1) \propto  (\Hb_{(1)}^{-1/2(t)} \gb_{(1)}^{(t)} - \Hb_{(1)}^{1/2(t)\top} \lb_1)^\top (\Hb_{(1)}^{-1/2(t)} \gb_{(1)}^{(t)} - \Hb_{(1)}^{1/2(t)\top} \lb_1),
\end{align}
where $\Hb_{(1)}^{1/2(t)}$ is any matrix such that $\Hb_{(1)}^{1/2(t)} \Hb_{(1)}^{1/2(t)\top} = \Hb_{(1)}^{(t)}$ and $\Hb_{(1)}^{-1/2(t)}$ is any generalized inverse of $\Hb_{(1)}^{1/2(t)}$.
It can be shown that $\gb_{(1)}^{(t)}$ is in the column space of
$\Hb_{(1)}^{(t)}$, and thus minimizing $-\Qcal_{(1)}(\Lb_1)$ with respect to
$\lb_1$ is equivalent to minimizing \eqref{eq:L1_quad_form} with
respect to  $\lb_1$. This statement remains true even when $\Hb_{(1)}^{(t)}$ is
a singular matrix. We further add to \eqref{eq:L1_quad_form} a group
lasso penalty as in Equation (10) in the main manuscript to enforce rank constraints
on our estimate of $\Lb_1$.  This completes the derivation of the
objective in Equation (10) of the main manuscript.

The objective for estimation $\Lb_2^{(t+1)}$ in Equation (11) of the main manuscript is
obtained by swapping indices 1 and 2 in the derivation above for $\Lb_1$ updates.

\section{MMTR Misspecification Plots}
\label{mmtr_misspec_plots}

We present additional plots for each of the simulation scenarios described in
Section 4.2 showing how the relative errors in estimating $\Lambdab$ by MMTR
change depending on the GEE's equicorrelation model parameter $\alpha$ (Figures
\ref{fig:misspec_fixed_scales} and \ref{fig:misspec_diff_scales}). The
$\alpha$ term in the GEE's equicorrelation model determines the correlation
between all pairs of responses within each group of observations. The greater
the $\alpha$ value, the more misspecified the model is to the MMTR setting, which
assumes the covariance between pairs of responses in each group to be a function
of that group's random-effects covariates. These plots demonstrate that the more
misspecified the correlation structure is from a random-effects setting, the
worse MMTR will perform in estimating the true covariance matrix.

\begin{figure}[H]
\includegraphics[width=\textwidth]{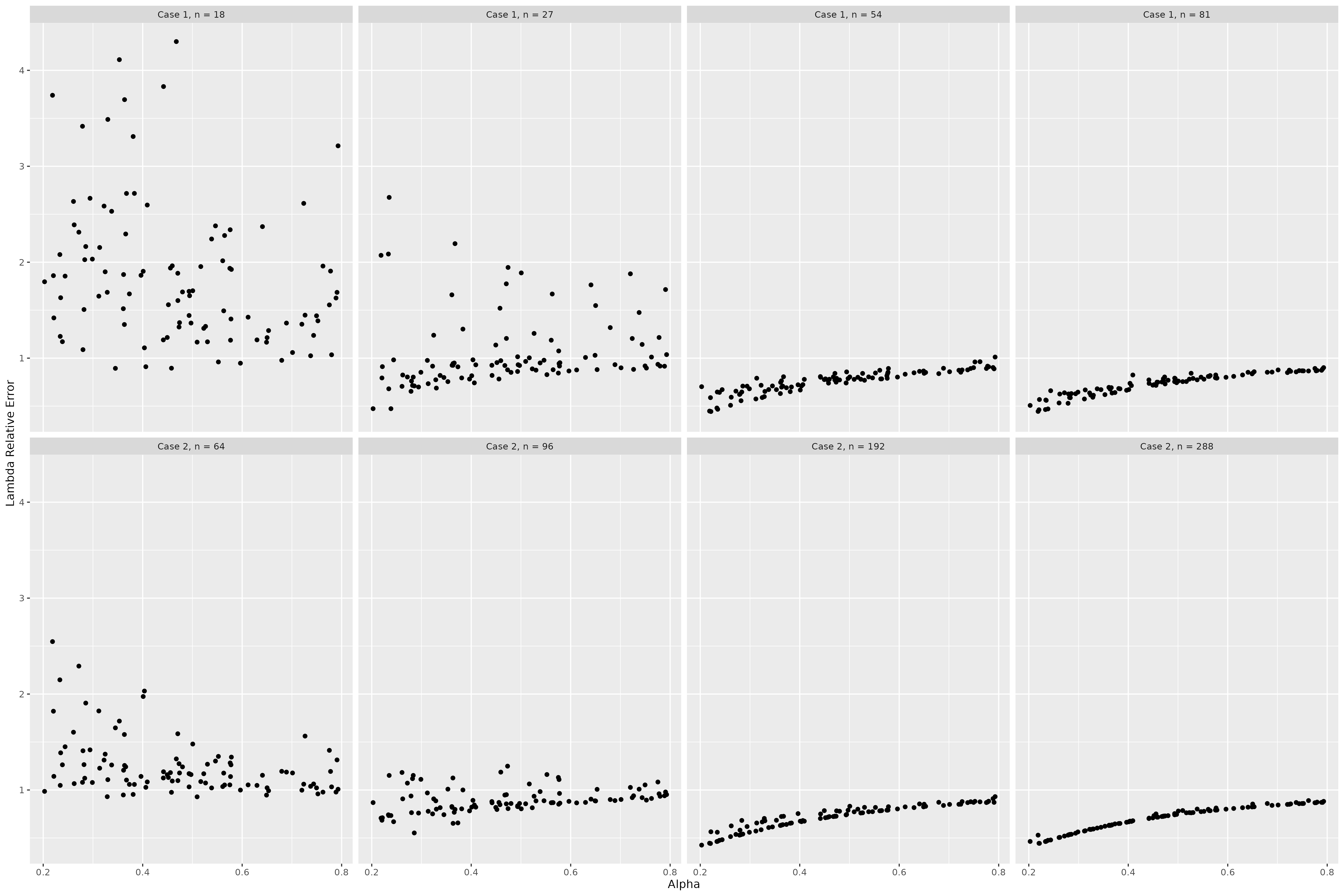}
\caption{Relative errors in estimating $\Lambdab$ by MMTR against $\alpha$. All plots here have the same range on the y axis.}
\label{fig:misspec_fixed_scales}
\end{figure}

\begin{figure}[H]
\includegraphics[width=\textwidth]{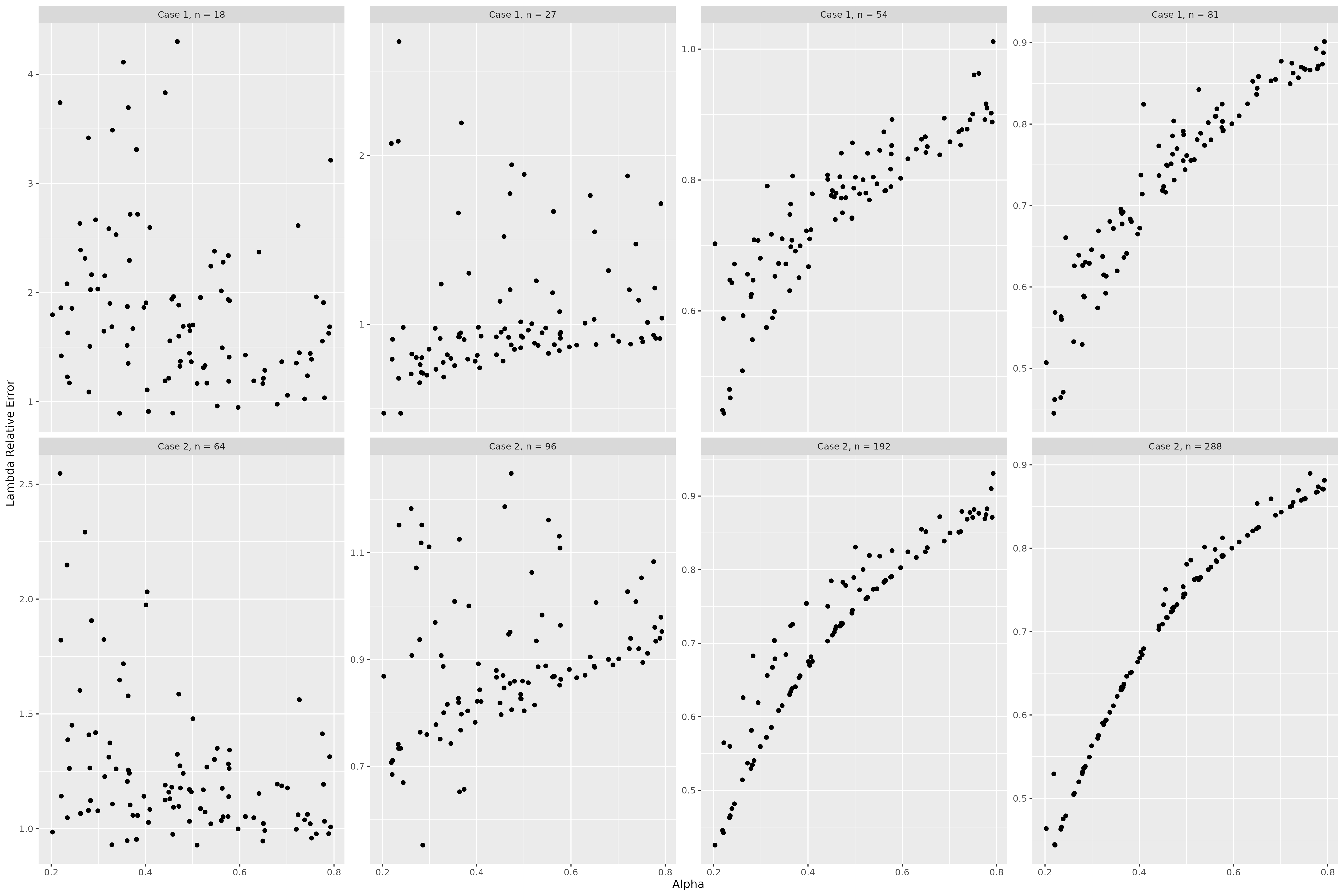}
\caption{Relative errors in estimating $\Lambdab$ by MMTR against $\alpha$. These plots are zoomed in from Figure \ref{fig:misspec_fixed_scales}.}
\label{fig:misspec_diff_scales}
\end{figure}

\end{document}